\documentclass{elsarticle}

\usepackage{fancybox}
\usepackage[all]{xy}
\usepackage{bbold}
\usepackage{graphicx}
\usepackage{bm}
\usepackage{color}
\usepackage{amsfonts}
\usepackage{amssymb}
\usepackage{amsthm}
\usepackage{amsmath}

\definecolor{red}{rgb}{1,0,0}
\definecolor{blue}{rgb}{0,0,1}
\definecolor{orange}{rgb}{1,0.5,0}

\newtheorem{thm}{Theorem}[section]
 \newtheorem{cor}[thm]{Corollary}
 \newtheorem{lem}[thm]{Lemma}
 \newtheorem{prop}[thm]{Proposition}
 \theoremstyle{definition}
 \newtheorem{defn}[thm]{Definition}
 \theoremstyle{remark}
 \newtheorem{rem}[thm]{Remark}
 \newtheorem*{ex}{Example}
 \numberwithin{equation}{section}
 \newtheorem{conj}[thm]{Conjecture}

\title{Phase space classification of an Ising Cellular Automaton: the Q2R model}

\author{Marco Montalva-Medel$^1$, Sergio Rica$^{1,2}$, and Felipe Urbina$^{3}$}
\address{ ${}^1$Facultad de Ingenier\'ia y Ciencias, Universidad Adolfo Ib\'a\~nez, Avda. Diagonal las Torres 2640, Pe\~nalol\'en, Santiago, Chile.\\
 ${}^2$ UAI Physics Center, Universidad Adolfo Ib\'a\~nez, Santiago, Chile.\\
  ${}^3$ Facultad de Ciencias,
Universidad Mayor,
Camino La Pirámide, 5750, Huechuraba, Santiago, Chile}

\begin{document}

\begin{abstract}
An exact characterization of the different dynamical behavior that exhibit the space phase of a reversible and conservative cellular automaton, the so called Q2R model, is shown in this paper. Q2R is a cellular automaton which is a dynamical variation of the Ising model in statistical physics and whose space of configurations grows  exponentially with the system size. As a consequence of the intrinsic reversibility of the model, the phase space is composed only by configurations that belong to a fixed point or a limit cycle. In this work we classify them in four types accordingly to well differentiated topological characteristics.
Three of them --which we call of type S-I, S-II and S-III-- share a symmetry property, while the fourth, which we call of type AS, does not.
Specifically, we prove that any configuration of Q2R belongs to one of the four previous limit cycles. Moreover, at a combinatorial level, we are able to determine the number of limit cycles for some small periods which are almost always present in the Q2R. Finally, we provide a general overview of the resulting decomposition of the arbitrary size Q2R phase space, in addition, we realize an exhaustive study of a small Ising system ($4\times 4$) which is fully analyzed under this new framework.
\end{abstract}

\maketitle


\section{Introduction}
A central problem in statistical physics concerns the manifestation of irreversibility whenever the system is governed by a large number of elements, or more precisely the number of degrees of freedom.  Despite the reversible 
character of the equation of motions in mechanics, the nature does not allow to observe a 
reversible behavior of a macroscopical system. After Boltzmann theory and the subsequent  Loschmidt and Zermelo's considerations this central question has 
been in the core of debates in basic physics since the end of the 19th century. In particular, Zermelo argued that the Boltzmann H-theorem is in contradiction with Poincar\'e's recurrence theorem, however,  as Boltzmann replied, the hypothetical recurrence time would be huge in comparison with all practical times in the usual thermodynamics. 

To model the recurrence time paradox, Paul and Tatiana Ehrenfest elaborated a particle exchange model 
\cite{Ehrenfest1,Ehrenfest2}, the so-called  the ``dog-flea'' model. This combinatorial model appears as an illustration of the irreversible exchange of heat between two distinct reservoirs at different temperatures.  Ehrenfest  model consists of $N$ particles that can be distributed in the left 
or right side of a container, in such a way that $N/2 + n$ balls are initially at the left 
container and $N/2-n$ in the right one. As shown by M. Kac \cite{Kac},  the  Ehrenfest  model may be 
mapped into a random walk. Moreover,  if initially the system is filling mostly one  container side $n\approx N/2$, then the average recurrence time is exponentially long, $\sim 2^N$. On the other hand, if initially the system 
is almost equally distributed $n\approx 0$, then, the waiting time scales as a diffusion process $\sqrt N$. Therefore, as one 
increases the total number of degrees of freedom $N$, there are some initial conditions with exponentially long recurrence times. 
The interest of the extremely simplified Ehrenfest model is that captures the essence of an exponentially long recurrence time showing that some initial configuration require exponentially long time to be back to the same state again. 

Generically, irreversibility arises as a consequence of systems that possess a  large number of degrees of freedom. 
Moreover, even in moderate system size, say $N=64$ for the  Ehrenfest's
model, the recurrence time becomes of the order of $2^{64}$, {\it i.e.,} essentially infinity. Therefore, although irreversibility appears to be a consequence of thermodynamic limit, $N\to\infty$, in practice  even 
 in moderate system size a thermodynamic statistical description appears to be the adequate one  \cite{SmallSystems}. 

In a recent article \cite{furbini}, two of us, developed a master equation approach to a reversible 
and conservative cellular automaton model: the Q2R model.  Introduced in the 80 by Vichniac \cite{vichniac}, Q2R is a dynamical variation of the Ising model for ferromagnetism that possesses quite a rich and complex dynamics. Remarkably, the evolution of Q2R preserves an Ising-like energy \cite{pomeau84}, appealing the analogy with the continuous dynamics of Hamiltonian systems\footnote{More details in Section \ref{sub-energy}.}. 
Because the Q2R model is a reversible cellular automaton its phase space is finite and there are neither attractive nor repulsive attractors, all attractors must be fixed points or limit cycles.  

Q2R is a two variable automaton, {\it i.e.}, a state is defined through $(x^t,y^t)$ in which each component $x^t$ and $y^t$ belong to a graph which is defined via a lattice and a neighbor (see Section \ref{model}). Although it can be defined in any kind of lattice, we restrict ourselves to the particular case of a square grid with a von-Neuman four nearest neighbors.
The size of the lattice will be $N= L\times L$, thus the phase space is the set of the $2^{2N}$ vertices of a $2N$-dimensional 
hypercube. However, as we show in this paper, the phase space is 
partitioned in a large number of subspaces composed by periodic orbits
or fixed points. 
A given initial condition 
belongs to one of this limit cycles or is a fixed point. 

It has been reported numerically, that the phase space is composed of a huge number of limit cycles with probable exponentially long periods \cite{hans2}. 
 For small Ising systems,  {\it e.g.}, for a $2\times 2$ square lattice, 
there are $2^8 = 256$ states and the longest orbit is of period 4. In the case, of a $4 \times 4 $, 
the phase space has $2^{32} \approx 4.3\times 10^9$ elements, being  
$T= 1080$ the longest limit cycle. More important, this case  can be scrutinized exactly, and we are able to conjecture that the 
number of states of a given period is exponentially large with the number of sites $N$.

In Ref. \cite{furbini}, following the Nicolis and Nicolis coarse-graining approach  \cite{nicolis}, we have applied it to the time series of the total magnetization, leading to a master equation that governs the macroscopic irreversible dynamics of the Q2R automata. The methodology works out for various lattice sizes. Notably, in the case of small systems, we show that the master equation leads to a tractable probability transfer matrix of moderate size, which provides a master equation for a coarse-grained probability distribution. The success of a consistent thermodynamic description is based on the existence of rich nature of the phase space. Similarly,  Lindgren and  Olbrich \cite{Lindgren} have recently considered the equilibrium properties of the Q2R model but with a different approach.
Furthermore, for a large system size it has been established that the evolution presents an irreversible 
behavior towards an equilibrium ruled by a micro-canonical ensemble \cite{hans,goles}. Moreover, in Ref.  \cite{goles}, it has been shown numerically that for 
a set of random initial conditions with different energies one recovers statistically the Ising phase transition 
ruled by the Onsager and Yang exact solutions \cite{onsager,yang}.

The aim of the present article, is to study and classify the different possible attractors (fixed points and limit cycles) of the phase space of the Q2R cellular automaton in a square lattice of arbitrary size. The starting point is the reversibility property of the Q2R model and essentially all results of the current paper follow after the Lemma \ref{LemRev} (on Reversibility).

Our main results are the following:

\begin{enumerate}
\item A fully classification of all attractors in four types of limit cycles consisting of symmetric and asymmetric ones (Theorem \ref{thm_cycle-clasiff}).
More precisely this characterization is  according to the specific topological features of the cycle. These limit cycles may be  symmetric limit cycle  of type S-I, S-II and S-III (See Sec. \ref{Secc:Cycles3andHigher}) and asymmetric limit cycle (AS).
\item The fixed points are of type S-I, moreover with the aid of splitting the lattice in two sub lattices we are able to show that the total number of fixed points is  $\beta^2 =4 k^2$ , with $k\in\mathbb{N}$ (Theorems \ref{MMM1} and \ref{teo:size-p1-p2}).
\item The characterization and existence of $ \beta^2 (\beta^2-1)$ period-two limit cycles. (Theorems \ref{MMM2} and \ref{teo:size-p1-p2}).
\item The characterization and existence of period-three limit cycles (Theorem \ref{teo_cl3} and Proposition \ref{prop:p3always}).
\end{enumerate}

The paper is organized as follows: In Section \ref{model}, we define the Q2R model and its main properties. In Section \ref{main}, we establish the formal definitions scheme, and we state the fundamental Lemma on reversibility (Lemma \ref{LemRev}) which is the key property after it all results in the paper follows. In Section \ref{sec_main}, we prove the main results listed above. Next, in Section \ref{SecOmega2}, we present a general overview of the resulting decomposition of the  phase space of Q2R. In Section \ref{Discussion} we conclude and discuss on further results and conjectures. Finally, in the Appendix \ref{AppEx} we provide an exhaustive study of a small Ising system ($4\times 4$) which is fully analyzed under this new framework with some specific examples of limit cycles.

\section{The model}\label{model}
\subsection{Context and definitions}
The Q2R model, introduced by Vichniac \cite{vichniac}, is defined in a regular two dimensional 
toroidal lattice with even rank $L\times L$, being $N=L^2$ the total number of \textit{nodes}\footnote{We focus our work with periodic boundary conditions on the lattice, but other possibilities may be also considered. In particular, the lattice does not require a square lattice. It could be a rectangular one: $L_1 \times L_2$.}  which have associated an index ${\bm k} \in  \{1,\dots ,N\} $, as well as a relative position in the lattice specified by two indices $k_1  \in\{1,\dots,L\}$  and $k_2 \in\{1,\dots,L\}$ (the respective row and column indices).  Further, a node $\bm k$ is characterized by two possible values $x_{\bm k}=\pm 1$, conforming with the following two-step rule:
$$
 x^{t+1}_{\bm k} = x^{t-1}_{\bm k} \, H \left( \sum_{{\bm i}\in V_{\bm k}} x^{t}_{\bm i} \right),
$$
where $V_{\bm k}$ denotes the von Neuman neighborhood of the four closest neighbors  with periodic boundary conditions. 
The function $H$ is such a that $H(s=0) = -1$ and $H(s) = +1$ in all other cases.

The above two-step rule may be naturally re-written as  a one step rule with the aid of a second dynamical variable \cite{pomeau84}:
\begin{eqnarray}
y^{t+1}_{\bm k} & = & x^{t}_{\bm k}  \nonumber \\
 x^{t+1}_{\bm k}  &= & y^{t}_{\bm k} \,  H\left( \sum_{i\in V_{\bm k} } x^{t}_{\bm i} \right).
 \label{q2r0}
\end{eqnarray}

Thus, the \textit{state} $x$ belongs to the discrete set 
$\Omega = \{ -1,1\} ^N$ (of size $2^N$) and
the set of \textit{configurations}, denoted by 
$\Omega^2$, it is composed by couples of states in 
$\Omega^2=\Omega\times\Omega=\left\{ (x,y)  / x\in \Omega \wedge y  \in \Omega \right\}$ (of size $2^{2N}$).

\begin{defn}
We denote the symbol $\odot $ by the Hadamard product, which is the multiplication component to 
component of the state $x\in \Omega$ and $y\in \Omega$. Hence, $x\odot  y\in \Omega$ represents that each 
component is defined by: $[x\odot y]_{ij} \equiv x_{ij} y_{ij}$.  This product is commutative, associative, 
and it possesses a neutral element, that we denote by $\mathbb{1}$ and corresponds to 
the state of $\Omega$ composed only by 1s. Moreover, we also define  $-\mathbb{1}\in\Omega$ by the states composed only by -1s. Given $x\in\Omega$, 
we will write $-x$ to refer to $-x=[-\mathbb{1}]\odot x$.
\end{defn}

\begin{defn}
Let be the function $\phi : \Omega \to \Omega$ such that, if $x\in \Omega$ then, the $ \bm k$-th 
component of $ [\phi(x) ]_{\bm k} = -1$ if the sum of all von Neuman neighbors of the $ \bm k$-th components is null, 
namely $\sum_{{\bm i} \in V_{\bm k} } x_{\bm i} =0$.
Notice that the neighborhood,  $V_{\bm k}$, includes  the periodic boundary condition of the lattice. Otherwise,  $[\phi(x) ]_{\bm k} =  +1$. Therefore, the 
function $\phi(x)$ is a state in $\Omega$ that has { a} -1 in the sites that $x$ has a null neighborhood. 
\end{defn}

\begin{ex}
Consider the state $x\in\Omega$ below. The node
(3,2) has null neighborhood 
 (its neighbors are marked by boxes), while the neighborhood of the node located at (1,4) 
is not null (its neighbors are marked by double boxes accordingly, to the toroidal lattice). So, the state $\phi (x)$ 
will have a -1 value at position (3,2) and a 1 value at position (1,4) that are also marked,   with a box 
and a double box, respectively, in $\phi (x)$. In a similar way,  all other values of $\phi(x)$ 
are obtained.
\begin{equation}\label{Eq:XandPhi}
x= \left[ \begin{array}{cccc} 
 \doublebox{1} & 1 & \doublebox{1} & 1 \\
1 & \fbox{-1} & 1 & \doublebox{1} \\
\fbox{1} & 1 &\fbox{-1} & 1 \\
1 & \fbox{1} & 1 & \doublebox{1} \\
\end{array} \right] \mapsto\phi (x) = \left[ \begin{array}{cccc} 
 1 & 1 & 1 & \doublebox{1} \\
1 & 1 & -1 & 1 \\
1 & \fbox{-1} & 1 & 1 \\
1 & 1 & 1 & 1 \\
\end{array} \right].
\end{equation}
\end{ex}

\begin{defn}\label{rem:phi1-1}
 The state $x$ does not have any null-neighborhood iff $\phi(x)=\mathbb{1}$.
Notice that $\phi(\mathbb{1})=\phi(-\mathbb{1})=\mathbb{1}$. 
\end{defn}

\subsection{The Q2R rule.}
Given $(x^t,y^t)\in \Omega^2$ at time $t$, and according with the previous definitions we re-write the Q2R model (\ref{q2r0}) as the 
following two step deterministic rule:

\begin{eqnarray}\label{q2r}
 y^{t+1}  &= & x^{t} \nonumber \\ 
 x^{t+1} &=& y^{t} \odot   \phi\left(x^{t}\right)  .
\end{eqnarray}

The \textit{evolution} is dictated by the rule (\ref{q2r}) and is complemented with an initial configuration 
$(x^{0},y^0) \in\Omega^2$.
For instance, let us consider $x^0=y^0=x$, the example given in (\ref{Eq:XandPhi}). The evolution of the initial configuration 
$(x^0,y^0)$ 
are obtained as follows:
\begin{center}
\scriptsize{
\begin{tabular}{ccc}
 $\left(\underbrace{\left[ \begin{array}{cccc} 
 1 & 1 & 1 & 1 \\
1 & -1 & 1 & 1 \\
1 & 1 & -1 & 1 \\
1 & 1 & 1 & 1 
\end{array} \right]}_{x^0}\right.$ &,&
$\left.\underbrace{\left[ \begin{array}{cccc} 
 1 & 1 & 1 & 1 \\
1 & -1 & 1 & 1 \\
1 & 1 & -1 & 1 \\
1 & 1 & 1 & 1
\end{array} \right]}_{y^0}\right)$\\  
&\xymatrix{
        \ar[dr] & \ar[dl]\\
         &}&\\
 $\left(\underbrace{\left[ \begin{array}{cccc} 
 1 & 1 & 1 & 1 \\
1 & -1 & -1 & 1 \\
1 & -1 & -1 & 1 \\
1 & 1 & 1 & 1
\end{array} \right]}_{x^1=y^0 \odot\phi(x^0)}\right.$&,&  $\left.\underbrace{\left[ \begin{array}{cccc} 
 1 & 1 & 1 & 1 \\
1 & -1 & 1 & 1 \\
1 & 1 & -1 & 1 \\
1 & 1 & 1 & 1
\end{array} \right]}_{y^1=x^0}\right)$\\
&\xymatrix{
        \ar[dr] & \ar[dl]\\
         &}&\\
\vdots & & \vdots         
\end{tabular}}\label{example2}
\end{center}
that we schematize with the following abbreviated notation: 
$$(x^0,y^0)\to (x^1,y^1)\to\cdots \, .$$
\begin{rem}\label{rem_iter}
In general, we will write $(x^t,y^t)\to (x^{t+1},y^{t+1}) $ for the one-step evolution from $(x^t,y^t)$ 
to $(x^{t+1},y^{t+1})=  \left(y^{t} \odot   \phi\left(x^{t}\right) ,x^{t} \right)$, according to rule (\ref{q2r}).
\end{rem}

\begin{defn}
The \textit{Phase Space} of the Q2R model it is composed by the set of configurations $\Omega^2$ and 
its one-step evolutions. Because is finite, the phase space has two types of {attractors}: \textit{limit cycles} or \textit{fixed points}. A limit cycle $\mathcal{C}$ of \textit{period} $T\in\mathbb{N}$ is a sequence dictated by the
evolution $(x^0,y^0)\to(x^1,y^1)\to\cdots\to (x^{T-1},y^{T-1})\to (x^{T},y^{T})$ such that all configurations $(x^t,y^t)$ are different, 
except $(x^0,y^0)=(x^{T},y^{T})$. We will write $(x,y)\in\mathcal{C}$ if $(x,y)$ is a configuration that is in $\mathcal{C}$ and, more general, the notation $[(x^t,y^t)\to(x^{t+1},y^{t+1})\to\cdots\to (x^{t+\tau},y^{t+\tau})]\in\mathcal{C}$ will be used to refer to the subsequence of $\mathcal{C}$ that goes from $(x^t,y^t)$ to $ (x^{t+\tau},y^{t+\tau})$, $\tau \leq T$. A fixed point is a limit cycle of period $T=1$, {\it i.e.,} is a configuration $(x,y)\in\Omega^2$ 
such that $(x,y)\to (x,y)$.
\end{defn}

\subsection{Main properties}
\subsubsection{Reversibility}

 Observe that $ \phi\left( x^{t}\right) \odot   \phi\left( x^{t}\right)=\mathbb{1}, \, \forall x^t \in \Omega$. 
Then Q2R rule may be inverted getting the backward evolution of the system but for the couple 
$(y^t , x^t )$, that reads:

\begin{eqnarray}\label{q2r.reversible}
\left \{ \begin{array}{l} x^{t-1}  =  y^{t} \\ y^{t-1}  =  x^{t} \odot   \phi\left( y^{t}\right), \end{array} \right. 
\end{eqnarray}
which is exactly the same rule (\ref{q2r}), displaying the remarkable property of reversibility. This property will be highlighted in the ``Reversibility Lemma'' (Lemma \ref{LemRev}).

\subsubsection{Configurations of the phase space}
 Since the Q2R rule is reversible and the phase space is finite, each configuration has two possibilities; to be a fixed point or to belong in a limit cycle. 
%

\subsubsection{Dynamical Invariants: The energy.}\label{sub-energy}
\begin{defn}
Let be the energy function
\begin{eqnarray}
E[\left(x^t,y^t\right)]  = -\frac{1}{2} \sum_{\left<{\bm i}, {\bm k}\right>}   x^{t}_{\bm i}   y^{t}_{\bm k} , 
\label{energy}
\end{eqnarray}
where the symbol $\left<\cdots\right>$ stands for sum over the four near neighbors. The energy (\ref{energy}) is bounded by $ -2N \leq E\leq 2 N$.
 \end{defn}

\begin{rem}\label{rem_energy}
As shown by Pomeau \cite{pomeau84}, the energy function (\ref{energy}) is conserved, under the dynamics defined by the Q2R rule (\ref{q2r}). That is $E[\left(x^t,y^t\right)] = E[\left(x^{0},y^{0}\right)] , \, \forall t \in \mathbb{N} $.
\end{rem}

\begin{rem}\label{rem_energy2}
Other dynamical invariants are known in the literature \cite{takesue2}.
\end{rem}

\subsubsection{The phase space and the qualitative dynamics}

The phase space of all configurations is defined through all possible values of the $2N$ 
dimensional state $(x,y) \in \Omega^2$. The resulting phase space is composed by the $2^{2N}$ vertices 
of a hypercube in dimension $2N$. This phase space  is partitioned in 
different sub-spaces accordingly to its energy, $E$, and accordingly with its dynamical characteristic such that the period, and other unknown parameters. 

For instance, the constant energy subspace shares in principle many limit cycles of different periods, as well as, many different fixed points. An arbitrary initial condition of energy $E$,  falls into one of these limit cycles, and it runs until a time  $T$, which could be exponentially long, and displaying a complex behavior (not chaotic \textit{stricto-sensu}, see for instance \cite{grassberger}). More important, the probability that an initial condition belongs to an exponentially long period limit cycle and it exhibits  
a complex behavior is finite \cite{hans2}. 
Moreover, Q2R  manifests sensitivity to initial conditions, that is, if one starts with two distinct, but 
close, initial conditions, then, they will evolve into very different limit cycles as time evolves \cite{goles}.  
In some sense, any initial state explores  vastly the phase space justifying the grounds of statistical 
physics, as we shown in the Remark \ref{rem:pxx<pxy} and the main Theorem \ref{thm_cycle-clasiff} on the limit cycle general classification.

\section{Preliminary Results}\label{main}
The core of this work aims to propose a new framework to study the dynamic of the Q2R model. It is 
based in particular properties of its configurations that allow to partition $\Omega^2$ in order 
to characterize the full spectrum of fixed points and limit cycles, as well as 
delving into topological and combinatorial aspects.
We begin by establishing the basic concepts that will be used along the text and the first results, 
necessary to understand the main results (Section \ref{sec_main}).
\subsection{Partitions of $\Omega^2$.}

Observe that, given the states $x,y\in\Omega$, there are two possibilities: $[x=y]$ or 
$[x\neq y]$.  Therefore a first partition of  $\Omega^2$ arises as follows:

\begin{defn}\label{defOmegaxx}
We denote by $\Omega^2_{xx}$ the \textit{set of configurations with equal states}, {\it i.e.,}
$$\Omega^2_{xx} = \left\{ (x,y)   \in \Omega^2 / x=y \right\},$$
whose size is $|\Omega^2_{xx}|=2^{N}$. Similarly, the \textit{set of configurations with different states} 
will be denoted by $\Omega^2_{xy}$ and corresponds to the complement set of $\Omega^2_{xx}$ in $\Omega^2$, 
{\it i.e.,}
$$\Omega^2_{xy} = \Omega^2 - \Omega^2_{xx}=\left\{ (x,y)   \in \Omega^2 / x\neq y \right\}$$
whose size is $|\Omega^2_{xy} |=2^N(2^N-1)$. Note that: $\Omega^2=\Omega^2_{xx}\cup\Omega^2_{xy}.$
\end{defn}

\begin{rem}\label{rem:pxx<pxy}
We underline that $|\Omega^2_{xx}|<|\Omega^2_{xy}|$. Further, the probability to have a configuration in $\Omega^2_{xx}$ is $p_{xx}= |\Omega^2_{xx}|/|\Omega^2|=  2^{-N}$, while the probability to have a configuration in $\Omega^2_{xy}$, is $p_{xy}= |\Omega^2_{xy}|/|\Omega^2|= 1-2^{-N}$. Hence, in practice, $p_{xx}\ll p_{xy}$. Moreover $p_{xy}\to 1$ in the thermodynamic limit ($N\to\infty$). 
\end{rem}

A second partition of $\Omega^2$ will allow us to know in detail the topology of the limit cycles: here, the sets $\Omega^2_{xx}$ and $\Omega^2_{xy}$ are also partitioned regarding the two possibilities of 
$\phi(x)$ in a configuration $(x,y)$, that is $[\phi(x)=\mathbb{1}]$ or $[\phi(x)\neq\mathbb{1}]$, regardless the value of 
$\phi(y)$, as in the following definition:

\begin{defn}\label{abcd}
Let be the following sets:
\begin{eqnarray}
A  &=& \left\{ (x,y)\in\Omega_{xx}^2  / \phi(x)=\mathbb{1} \right\}\nonumber\\
B  &=& \left\{ (x,y)\in\Omega_{xy}^2  / \phi(x)=\mathbb{1} \right\}\nonumber\\
C  &=& \left\{ (x,y)\in\Omega_{xx}^2  /  \phi(x)\neq\mathbb{1} \right\}\nonumber\\
D  &=& \left\{ (x,y)\in\Omega_{xy}^2  / \phi(x)\neq\mathbb{1} \right\}.\nonumber
\end{eqnarray}
We say that $(x,y)\in\Omega^2$ is a \textit{configuration of type} $A$, $B$, $C$ or $D$, if $(x,y)$ belongs to one of the sets $A$, 
$B$, $C$ or $D$, respectively. Later on, we refer to a \textit{evolution} of type $U\to V$ to the one step evolution of a configuration $(x,y)\in U$ up to $(w,x)\in V$
with $U,V\in\{A,B,C,D\}$. In such a case, we will say that $U$ \textit{evolves} to $V$, or $V$ 
\textit{comes from} $U$ (see Figure \ref{fig_transitions} and Corollary \ref{cor_c3props} 
as a given examples of this terminology).
\end{defn}

\begin{defn}
We denote by $P_T\subset \Omega^2$ the \textit{set of configurations belonging to a limit cycle of period $T\in\mathbb{N}$}. 
In particular, $P_1$ correspond to the set of fixed points of Q2R.
 Moreover, we will denote by $\nu(T)$ the size of the set $P_T$ and by $n(T)$ the number of limit cycles of period $T$. That is: $\nu(T)\equiv |P_T|$ and $n(T)\equiv\frac{\nu(T)}{T}$. 
\end{defn}

\begin{rem}
 From definition \ref{abcd}, $\Omega^2_{xx} =A\cup C$ and $\Omega^2_{xy} = B\cup D$. 
Further, we show in Remarks \ref{rem_p1} and \ref{rem_p2}  that $P_1 = A\subset \Omega^2_{xx}$ and $P_2 \subset B\subset\Omega^2_{xy}$, respectively.
\end{rem}

\begin{defn}
We say that, the  \textit{symmetric}  configuration of $(x,y)\in\Omega^2$ is the configuration $(y,x)\in\Omega^2$. In particular, 
the symmetric configuration of $(x,x)\in\Omega^2$ is itself, $(x,x)$, and we will call it as a \textit{self-symmetric} configuration. 
We say that a limit cycle $\mathcal{C}$ is \textit{symmetric} if satisfy:
$$ (x,y)\in\mathcal{C}\Rightarrow (y,x)\in \mathcal{C}.$$
Otherwise, we say that $\mathcal{C}$ is \textit{non-symmetric}.
\end{defn}

Naturally, the above definition allow us to separate all the attractors of Q2R into symmetric and 
non-symmetric, however, from our (main) Theorem \ref{thm_cycle-clasiff}, it will be shown that the 
symmetric ones are of 3 types, while that the non-symmetric ones possess the following (stronger) property defined below (see Figure \ref{CycleDraft}).

\begin{defn}\label{NonSymmCycle}
A non-symmetric limit cycle $\mathcal{C}$ is said to be \textit{asymmetric} if satisfy: 
$$ (x,y)\in\mathcal{C}\Rightarrow (y,x)\notin \mathcal{C}.$$
\end{defn}

 \begin{figure}[h]
\begin{center}
 \includegraphics[width = 4cm]{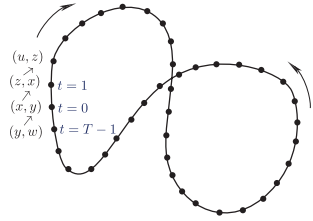} 
\end{center}
\caption{ \label{CycleDraft}  Scheme of an asymmetric limit cycle $\mathcal{C}$; if a configuration $(x,y)$ belongs to $\mathcal{C}$, then $(y,x)$ does not.}
\end{figure}

Next, we continue with a key property of the Q2R model that will allow to have an easy understanding 
of the attractor classification shown in this paper.

\subsection{Fundamental property between configurations and its symmetrical}
The following Reversibility Lemma shows a main characteristic of the Q2R
system (\ref{q2r}).

\begin{lem}[Reversibility]\label{LemRev} 
Let $x$, $y$, $z$ in $\Omega$, then,
$$ [\left( x, y \right) \to   \left(z, x\right)] \Leftrightarrow
[\left(x, z\right)  \to   \left( y, x\right)].$$
\end{lem}
\begin{proof}
From (\ref{q2r}) and because of $ \phi\left( x\right) \odot   \phi\left( x\right)=\mathbb{1}$:
  \begin{eqnarray}
  [\left( x, y \right) \to   \left(z, x\right)] & \Leftrightarrow & \left\{ \begin{array}{l} x  =  x \\ z  =  y \odot   \phi\left( x\right) \end{array} \right.\nonumber\\
  & \Leftrightarrow & \left\{ \begin{array}{rl} x  &=  x \\ z \odot   \phi\left( x\right) &=  y \end{array} \right.\nonumber\\
  & \Leftrightarrow & [\left(x, z\right)  \to   \left( y, x\right)]\nonumber
  \end{eqnarray}
\end{proof}

This Reversibility Lemma says that, if there is a one time step evolution between two configurations, then, 
there is also a one time step evolution between their symmetric configuration, but, in the opposite sense.  As a consequence, we have the following generalization:

\begin{cor}\label{cor_sec-sim-trans}
Let $x^t$, $y^t$ in $\Omega$, $t\in\{0,...,p\}$, $p\in\mathbb{N}$, then,
$$ \left( x^0, y^0 \right) \to\cdots\to (x^{p},y^{p}) \Leftrightarrow
(y^{p},x^{p})\to\cdots\to \left( y^0, x^0 \right).$$
\end{cor}

Figure \ref{fig_sec-sim-trans} illustrates the proof of this property.

\begin{figure}[h]
\begin{center}
$$
\begin{array}{cccc}
  \underbrace{  (x^0,y^0)\to \left( x^1, y^1 \right) } &  \to & \cdots  & \underbrace{ \to (x^{p},y^{p})  } \\
  {\Updownarrow}  & & \cdots & \Updownarrow \\
  \overbrace{  (y^0,x^0)\gets \left( y^1, x^1 \right) } &   \gets &\cdots & \overbrace{ \gets (y^{p},x^{p})  } $$
\end{array}
$$
\end{center}
\caption{Applying successively the Lemma \ref{LemRev}
at each step-evolution $(x^t,y^t)\to (x^{t+1},y^{t+1})$,  for $t\in\{0,\dots,{p}\}$, $p\in\mathbb{N}$, one constructs the backward evolution between their symmetric configurations.}\label{fig_sec-sim-trans}
\end{figure}

Let us study the possible evolutions, according to the type of configurations involved.

\subsection{The possible evolutions between configurations of type $A$, $B$, $C$ and $D$}

Given a configuration of type $U\in\{A,B,C,D\}$, then the only possible evolutions are:
\begin{itemize}
\item[(T1)] Configurations of type $A$.\\ 
Let $(x,y)\in A$, then {$[x=y] \wedge [\phi(x)=\mathbb{1}]$}. Since $\phi(x)=\mathbb{1} $, then $(x,y)\to (y,x)$, and because $y=x$, then $\phi(y)=\mathbb{1}$. Therefore, 
$(x,x)\to (x,x)\in A$. In fact, this is the characterization of the fixed points of Q2R (Theorem \ref{MMM1}). 
Thus, $A\to A$ as shown in Figure \ref{fig_transitions}-(T1). 
\item[(T2)] Configurations of type $B$.\\ 
Let $(x,y)\in B$, then $[x\neq y] \wedge [\phi(x)=\mathbb{1}]$. Because $\phi(x)=\mathbb{1} $, then $(x,y)\to (y,x)$. Hence, there are two possibilities for $\phi(y)$:
\begin{itemize}
\item[(i)] If $\phi(y)=\mathbb{1}$, then $(x,y)\to (y,x)\in B$. In fact, this is the characterization of the limit cycles of period two of  the Q2R  model (Theorem \ref{MMM2}). 
\item[(ii)] If $\phi(y)\neq\mathbb{1}$, then $(x,y)\to (y,x)\in D$. 
\end{itemize}
Thus, $[B\to B]$ or $[B\to D]$, as shown in Figure \ref{fig_transitions}-(T2).
\item[(T3)] Configurations of type $C$.\\ 
Let $(x,y)\in C$, then $[x=y] \wedge [\phi(x)\neq\mathbb{1}]$. Hence,   $(x,y)\to(z=y\odot \phi(x),x)$ with $z\neq y$ (consequently, $z\neq x$) and there are two possibilities for $\phi(z)$:
\begin{itemize}
\item[(i)] If $\phi(z)=\mathbb{1}$, then $(x,y)\to (z,x)\in B$. 
\item[(ii)] If $\phi(z)\neq\mathbb{1}$, then $(x,y)\to (z,x)\in D$.
\end{itemize}
Thus, $[C\to B]$ or $[C\to D]$, as shown in Figure \ref{fig_transitions}-(T3).
\item[(T4)] Configurations of type $D$.\\ 
Let $(x,y)\in D$, then $[x\neq y] \wedge [\phi(x)\neq\mathbb{1}]$. Hence, $(x,y)\to(z=y\odot \phi(x),x)$ with $z\neq y$ (but eventually, $z=x$) and there are two possibilities for $z$, which implies three possible evolutions for $(x,y)$:
\begin{itemize}
\item[(i)] If $z=x$, then $(x,y)\to (z,x)=(x,x)\in C$. 
\item[(ii)] If $z\neq x$, then we have the same two possibilities (T3)-(i) and (T3)-(ii) for $\phi(z)$.
\end{itemize}
Therefore, $[D\to B]$, $[D\to C]$ or $[D\to D]$, as shown in Figure \ref{fig_transitions}-(T4).
\end{itemize}
The above analysis is summarized in Figure \ref{fig_transitions}.

\begin{figure}[h]
\begin{center}
\begin{tabular}{cccc}
(T1) \includegraphics[width = 1. cm]{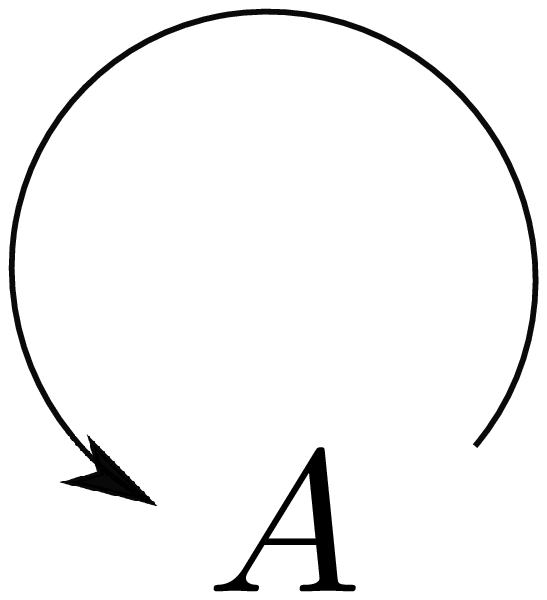} & (T2) \includegraphics[width = 1.5 cm]{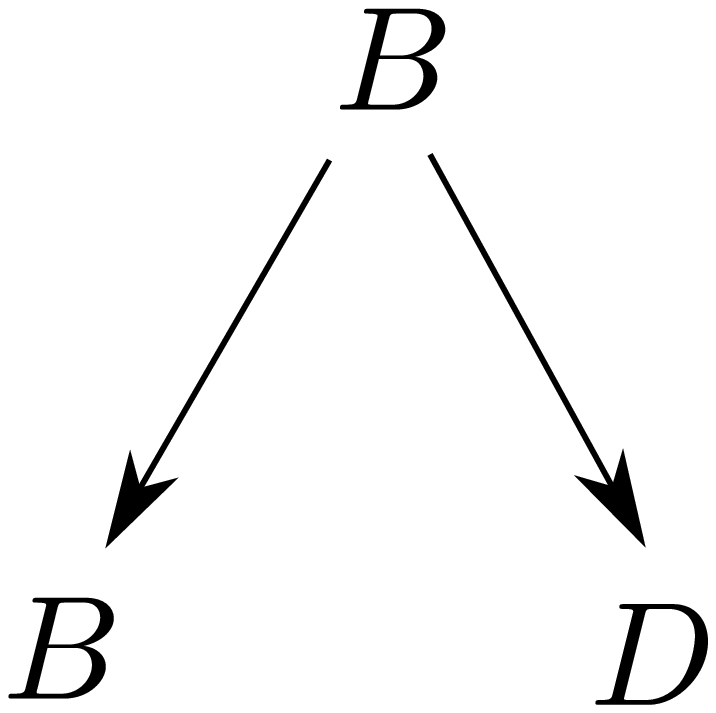} 
& (T3) \includegraphics[width = 1.5 cm]{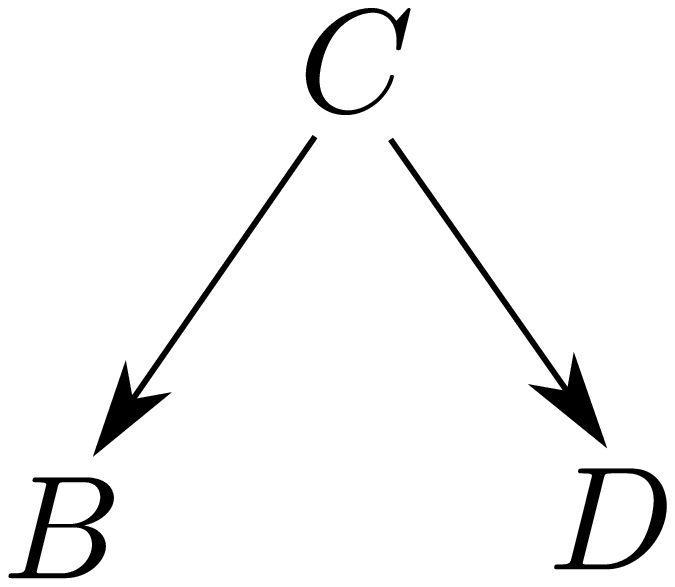} 
& (T4) \includegraphics[width = 1.5 cm]{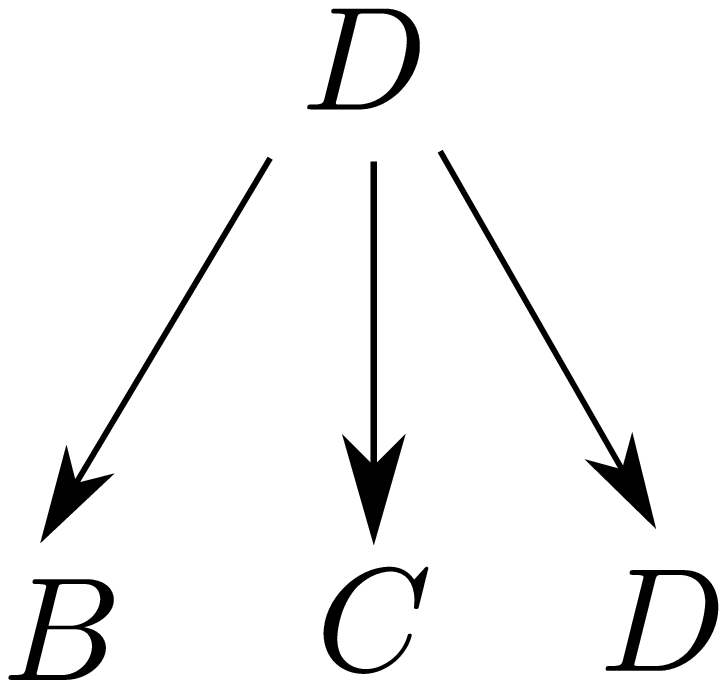}
\end{tabular}
\end{center}
\caption{(T1) A type $A$ configuration only can evolve to itself. (T2)-(T3) Configurations 
of type $B$ and $C$ can evolve to configurations of type $B$ or $D$. (T4) A type $D$ configuration can evolve 
to any configuration, excepting the configurations of type $A$.}
\label{fig_transitions}
\end{figure}

\section{Main Results}\label{sec_main}
\subsection{Characterization of fixed points and limit cycles of period two and higher.}
\begin{thm}[Characterization of fixed points] \label{MMM1}

Let $(x,y)\in\Omega^2$ be a configuration of the Q2R model. Then,  
$$(x,y)\in P_1 \Leftrightarrow [x=y]\wedge[\phi(x)=\mathbb{1}].$$
\end{thm}
\begin{proof}  
\begin{eqnarray}
(x,y)\in P_1 & \Leftrightarrow & [(x,y)  \to (x,y)] \text{, by definition of fixed point.}\nonumber\\
& \Leftrightarrow & \left \{ \begin{array}{rl} x  &=  y \\ x &=  y \odot  \phi(x)\end{array} \right.\text{, by Remark \ref{rem_iter}.}\nonumber\\
& \Leftrightarrow & [x=y]\wedge [\phi(x)=\mathbb{1}].\nonumber
\end{eqnarray}
\end{proof}

\begin{rem}\label{rem_p1}
The above result states that fixed points of the Q2R model are always configurations of type $A$ (see Figure \ref{fig_fp-p2}-a), {\it i.e.}, 
 self-symmetric ($(x,x)\in \Omega^2_{xx}$) and without null neighborhoods ($\phi(x)= \mathbb{1}$). 
Hence, $P_1=A$.
\end{rem}

\begin{figure}[h]
\begin{center}
a)\quad  \includegraphics[width = 1.5 cm]{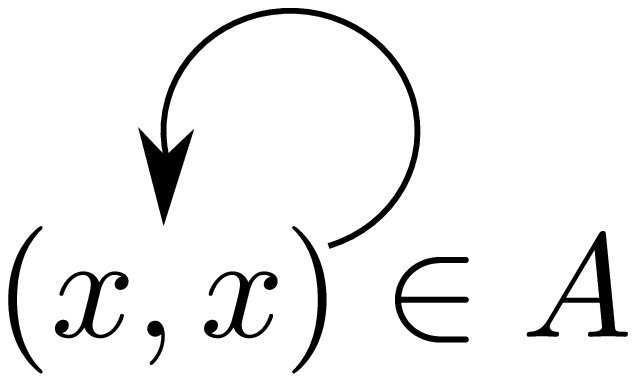} \quad \quad \quad b) \quad  \includegraphics[width = 1.5 cm]{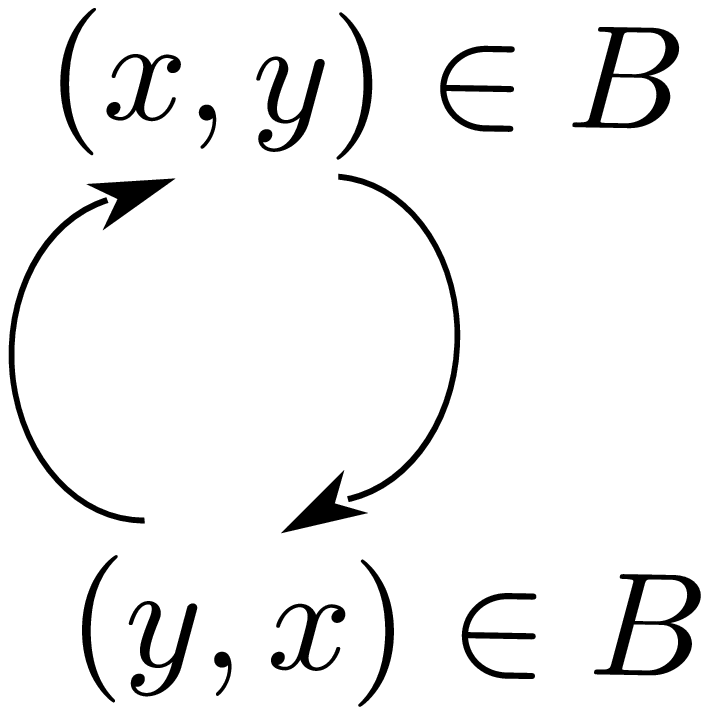} 
\end{center}
\caption{a) Scheme for a fixed point (or limit cycle of period 1). b) Scheme for a period-2 limit cycle. }
\label{fig_fp-p2}
\end{figure}

Since Q2R is a reversible system, if a configuration $(x,y)\in\Omega^2$ is not a fixed point, then necessarily it belongs to a limit cycle of period  2 or higher. In this context, as a characterization of 
such a limit cycles, we consider convenient
to explicit the next Corollary, which is the negation of Theorem \ref{MMM1}.

\begin{cor}\label{coro_cl}
Let $(x,y)\in\Omega^2$ be a configuration of Q2R. Then,
$$(x,y) \text{ belongs in a limit cycle of period 2 or higher} \Leftrightarrow [x\neq y] \vee [\phi(x) \neq\mathbb{1}].$$
\end{cor}

\begin{rem}
The possible evolutions analyzed before implies that the configurations involved in any limit cycle of period 2 or higher are of type $B$, $C$ or $D$ (not A).
\end{rem}

\subsection{Period 2 limit cycles}
\begin{thm}[Characterization of limit cycles of period 2] \label{MMM2}

Let $(x,y)\in\Omega^2$ be a configuration of Q2R. Then,
$$(x,y)\in P_2 \Leftrightarrow [x\neq y] \wedge [\phi(x)=\mathbb{1}] \wedge [\phi(y)=\mathbb{1}].$$
\end{thm}
\begin{proof}
\begin{eqnarray}
(x,y)\in P_2 & \Leftrightarrow & [(x,y)  \to (y\odot \phi(x),x) \to (x,y)]\text{, by definition of } P_2.\nonumber\\
& \Leftrightarrow & [(x,y)  \to (y\odot \phi(x),x)] \wedge \begin{array}{c} \left \{ \begin{array}{rl} y  &=  y\odot \phi(x) \\ 
x &=  x \odot  \phi(y\odot \phi(x))\end{array} \right.\\
\text{(by Remark \ref{rem_iter})}\end{array}\nonumber\\
& \Leftrightarrow & [(x,y)  \to (y\odot \phi(x),x)] \wedge \left \{ \begin{array}{rl} \phi(x)  &= \mathbb{1} \\ 
\phi(y\odot \phi(x)) &= \mathbb{1} \end{array} \right.\nonumber\\
& \Leftrightarrow & [(x,y)  \to (y,x)] \wedge \left \{ \begin{array}{rl} \phi(x)  &= \mathbb{1} \\ 
\phi(y) &=  \mathbb{1}  \end{array} \right.\nonumber\\
& \Leftrightarrow & [x\neq y] \wedge [\phi(x)= \mathbb{1}] \wedge [\phi(y) = \mathbb{1}].\nonumber
\end{eqnarray}
\end{proof}

\begin{rem}\label{rem_p2}
The above result says that the limit cycles of period 2 consists of configurations $(x,y)$, 
of type $B$ and with both states, $x$ and $y$, without null neighborhoods (see Figure \ref{fig_fp-p2}b). 
Hence, $P_2\subset B$. 

 Observe that there are elements in $B$ which do not belong to $ P_2$. For instance, take the configuration $(\mathbb{1},x)\in B$ where the state $x\in\Omega$ is composed by a $2\times 2$ block of -1s surrounded by 1s, {\it i.e.,}

$$
x=\left[
\begin{array}{cccccccc}
1 & 1 & \cdots & \cdots&  \cdots & \cdots & 1 & 1 \\
1 & \ddots & \cdots & \cdots&  \cdots & \cdots &\vdots & 1 \\
1 & \ddots & 1 & 1&  1 & 1 &\cdots & 1 \\
\vdots & \cdots &1 & -1 & -1 & 1& \cdots  &\vdots \\
\vdots & \cdots & 1& -1 & -1 & 1&\cdots & \vdots \\
1 & \cdots & 1 & 1&  1 & 1 &\cdots & 1 \\
\vdots & \vdots & \cdots & \cdots & \cdots & \cdots & \ddots & 1 \\
1 & 1 & \cdots & \cdots&  \cdots & \cdots & 1 & 1 
\end{array}
\right]_{L\times L},
$$
This configuration has only 4 null neighborhoods, located just in the nodes of the inner block of -1s. In other words, $\phi (x) = x$. Thus, by applying the Q2R rule  we {get}~: 
\begin{itemize}
\item $(\mathbb{1},x)\to (\phi(\mathbb{1})\odot x,\mathbb{1})=(x,\mathbb{1})$
\item $(x,\mathbb{1})\to (\phi(x)\odot \mathbb{1},x)=(x\odot \mathbb{1},x)=(x,x)$
\item $(x,x)\to (x\odot \phi(x),x)=(x\odot x,x)=(\mathbb{1},x)$
\end{itemize}
Therefore, the sequence of evolutions $(\mathbb{1},x)\to(x,\mathbb{1})\to(x,x)\to(\mathbb{1},x)$ is a period-3 limit cycle and, consequently, $(\mathbb{1},x)\in P_3$, hence, $(\mathbb{1},x)\notin P_2$.
\end{rem}

\subsection{Period 3 limit cycles}
\begin{thm}[Characterization of limit cycles of period 3]\label{teo_cl3} 
Let $\{(x,y),(z,x),(y,z)\}\subset\Omega^2$ such that $(x,y)\to(z,x)\to(y,z)$. Then,
$$\{(x,y),(z,x),(y,z)\}\subseteq P_3 \Leftrightarrow \phi (x)\odot \phi (y)\odot \phi (z)=\mathbb{1}.$$
\end{thm}
\begin{proof}
By Remark \ref{rem_iter}, $(x,y)\to(z,x)\to(y,z)$ means that:
\begin{eqnarray*}
 \underbrace{[y=x\odot \phi (z)]}_{(a)} \wedge \underbrace{[z=y\odot \phi (x)]}_{(b)} 
\end{eqnarray*}
$\Rightarrow )$
\begin{eqnarray}
\{(x,y),(z,x),(y,z)\}\subseteq P_3 & \Rightarrow & [(y,z) \to (x,y)] \nonumber\\
& \Rightarrow & x=z\odot \phi (y) \text{, by Remark \ref{rem_iter}.}\label{eq-cl3_2}
\end{eqnarray}
Replacing (a) in (b) and, after that, (b) in (\ref{eq-cl3_2}) we have that $x\odot \phi (x)\odot \phi (y)\odot \phi (z)=x$, {\it i.e.,}
$$\phi (x)\odot \phi (y)\odot \phi (z)=\mathbb{1}.$$

\noindent $\Leftarrow )$ 
\begin{eqnarray}
\phi (x)\odot \phi (y)\odot \phi (z)=\mathbb{1} & \Rightarrow & x\odot \phi (x)\odot \phi (y)\odot \phi (z)=\mathbb{1}\odot x=x \nonumber\\
& \Rightarrow & [x\odot \phi (z)]\odot \phi (x)\odot \phi (y)=x\nonumber\\
& \Rightarrow & y\odot \phi(x)\odot \phi (y)=x \text{, by { (a)}.}\nonumber\\
& \Rightarrow & z\odot \phi (y)=x \text{, by { (b)}.}\nonumber\\
& \Rightarrow & [(y,z) \to (x,y)] \nonumber
\end{eqnarray}
Thus, $[(x,y)\to(z,x)\to(y,z)]\wedge [(y,z) \to (x,y)]$, {\it i.e.,}
$$\{(x,y),(z,x),(y,z)\}\subseteq P_3$$
\end{proof}

\begin{rem}
Contrary with the statements of the Remarks \ref{rem_p1} and \ref{rem_p2}, in which fixed points and period-2 
limit cycles  are shown to have a unique topology, the condition $\phi (x)\odot \phi (y)\odot \phi (z)=\mathbb{1}$ 
for period-3 limit cycles allows different topologies (see Figure \ref{Period3}). This fact occurs for all limit cycles {of} period 3 and higher and it will be explained later in Remark \ref{rem_cyclededuction}.
\end{rem}

 \begin{figure}[htbp!]
\begin{center}
a)\quad \includegraphics[width=3.5cm]{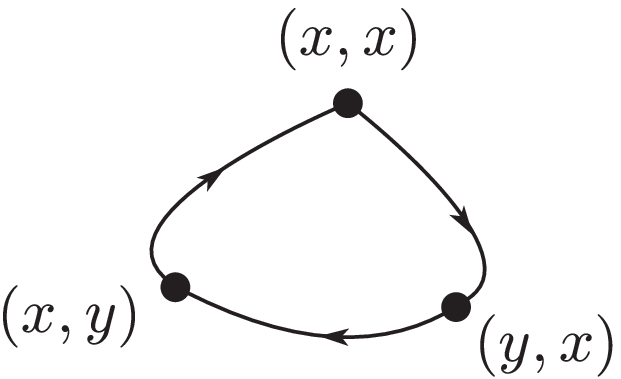}  \quad\quad \quad b)\quad\includegraphics[width=3.5cm]{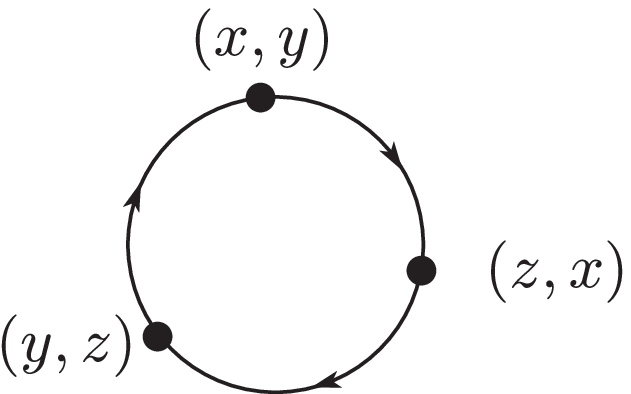} 
\end{center}
  \caption{\label{Period3} The two different topologies for a period-3 limit cycle: a) A symmetric limit cycle, 
  with one self-symmetric configuration $(x,x)$. b)  An asymmetric limit cycle, hence, without any self-symmetric configuration. }
\end{figure}

\subsection{Properties of the configurations in limit cycles with period three and higher.}\label{Secc:Cycles3andHigher}
The following result will be useful for the proof of our main Theorem \ref{thm_cycle-clasiff} 
(on the attractors classification in Q2R) and dictates direct consequences 
that are easily deduced from Theorems \ref{MMM1} and \ref{MMM2} and the possible evolutions 
shown in Figure \ref{fig_transitions}.
\begin{cor}\label{cor_c3props}
Let $\mathcal{C}$ be a limit cycle of Q2R with period 3 or higher. Then:
\begin{itemize}
\item[(i)] {$\mathcal{C}$ has at least one configuration of type $D$.}
\item[(ii)] $\mathcal{C}$ does not have evolutions of type $A\to A$, nor  $B\to B$ (notice that  $C\to C$ does not exist) but it could have evolutions of type $D\to D$.
\item[(iii)] If $\mathcal{C}$ has a type $D$ configuration, then $D$ comes from a type $V$ configuration {with} $V\in\{B,C,D\}$. 
\item[(iv)] If $\mathcal{C}$ has a configuration $(x,y)\in B$, then $(x,y)\to (y,x)\in D$.
\end{itemize}
\end{cor}

\subsection{Topological classification of limit cycles in Q2R}\label{subsec_typeLC}
\begin{defn}
Let $\mathcal{C}$ be a limit cycle of the Q2R  model with period $T\in\mathbb{N}$. We say that $\mathcal{C}$ 
is:
\begin{itemize}
\item A symmetric limit cycle of type I (S-I). If $T=1$ or if there exists $p\in\mathbb{N}_0$ such that $\mathcal{C}$  has the topology of Figure \ref{fig_cycletype}-a, {\it i.e.}, is symmetric with:
\begin{itemize}
\item An odd period $T=2(p+1)+1$.
\item Only one configuration of type $C$, only one configuration of type $B$  and $(2p+1)$ configurations of type $D$.
\end{itemize}
\item A symmetric limit cycle of type II (S-II). If there exists $p\in\mathbb{N}_0$ such that $\mathcal{C}$ has the topology of Figure \ref{fig_cycletype}-b, {\it i.e.}, is symmetric with:
\begin{itemize}
\item An even period $T=2(p+2)$.
\item  Only two configurations of type $C$ and $2(p+1)$ configurations of type $D$.
\end{itemize}
\item {A symmetric limit cycle of type III (S-III).} If $T=2$ or if there exists $p\in\mathbb{N}_0$ such that $\mathcal{C}$ has the topology of Figure \ref{fig_cycletype}-c, {\it i.e.}, is symmetric with:
\begin{itemize}
\item {An even period $T=2(p+2)$}.
\item Only two type $B$ configurations and $2(p+1)$ type $D$ configurations.
\end{itemize}
\item {An asymmetric limit cycle (AS).} If there exists $p\in\mathbb{N}\setminus\{1\}$ such that $\mathcal{C}$ has the topology of one
 of the two limit cycles of Figure \ref{fig_cycletype}-d, {\it i.e.}, is an asymmetric limit cycle with:
\begin{itemize}
\item Period $T=p+1$ (it can be even or odd, depending on the value of $p$).
\item All its configurations are of type $D$.
\end{itemize}
\end{itemize}
\end{defn}

\begin{figure}[h]
\begin{center}
a) \includegraphics[width = 4 cm]{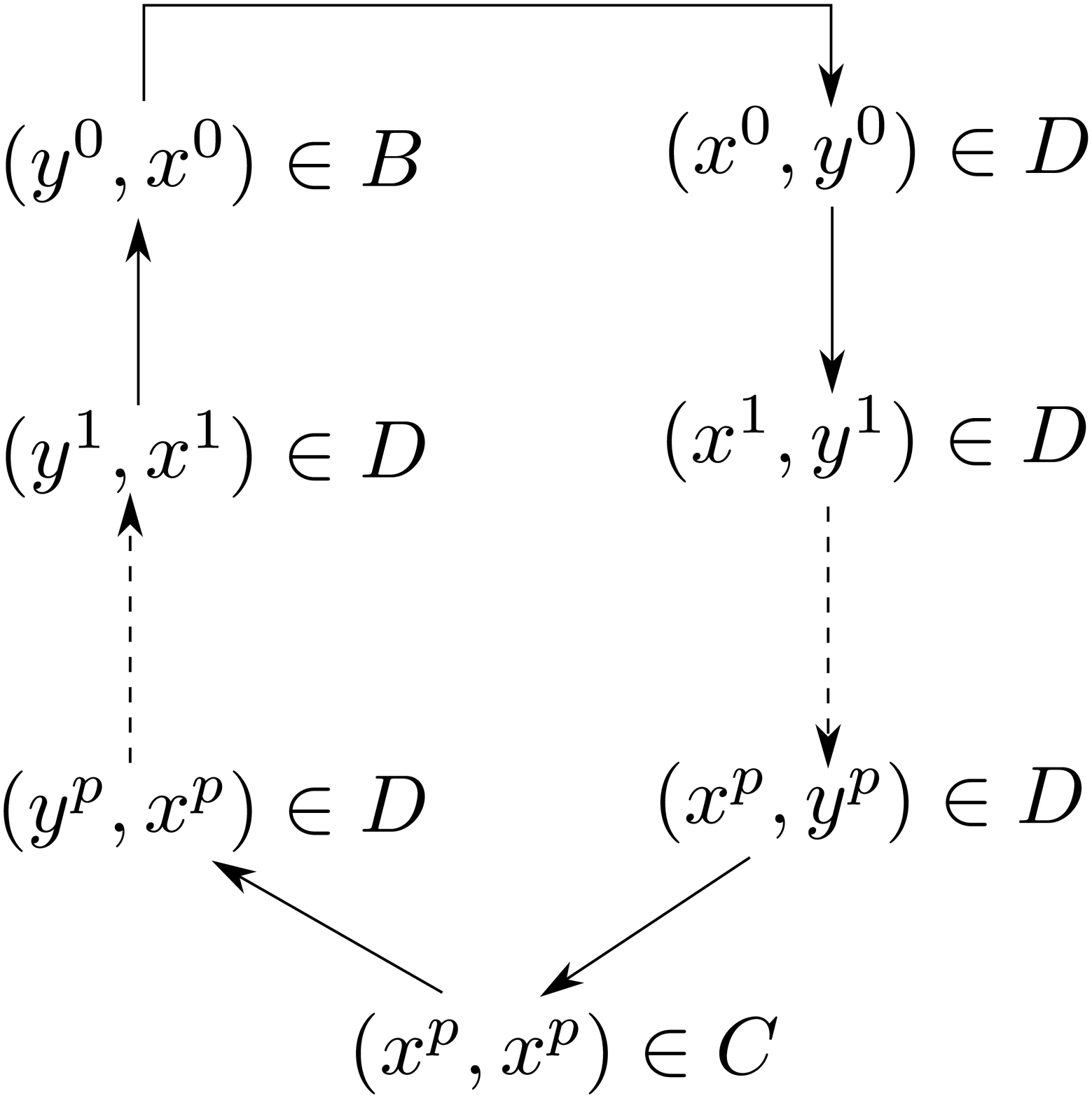} \quad b) \includegraphics[width = 4 cm]{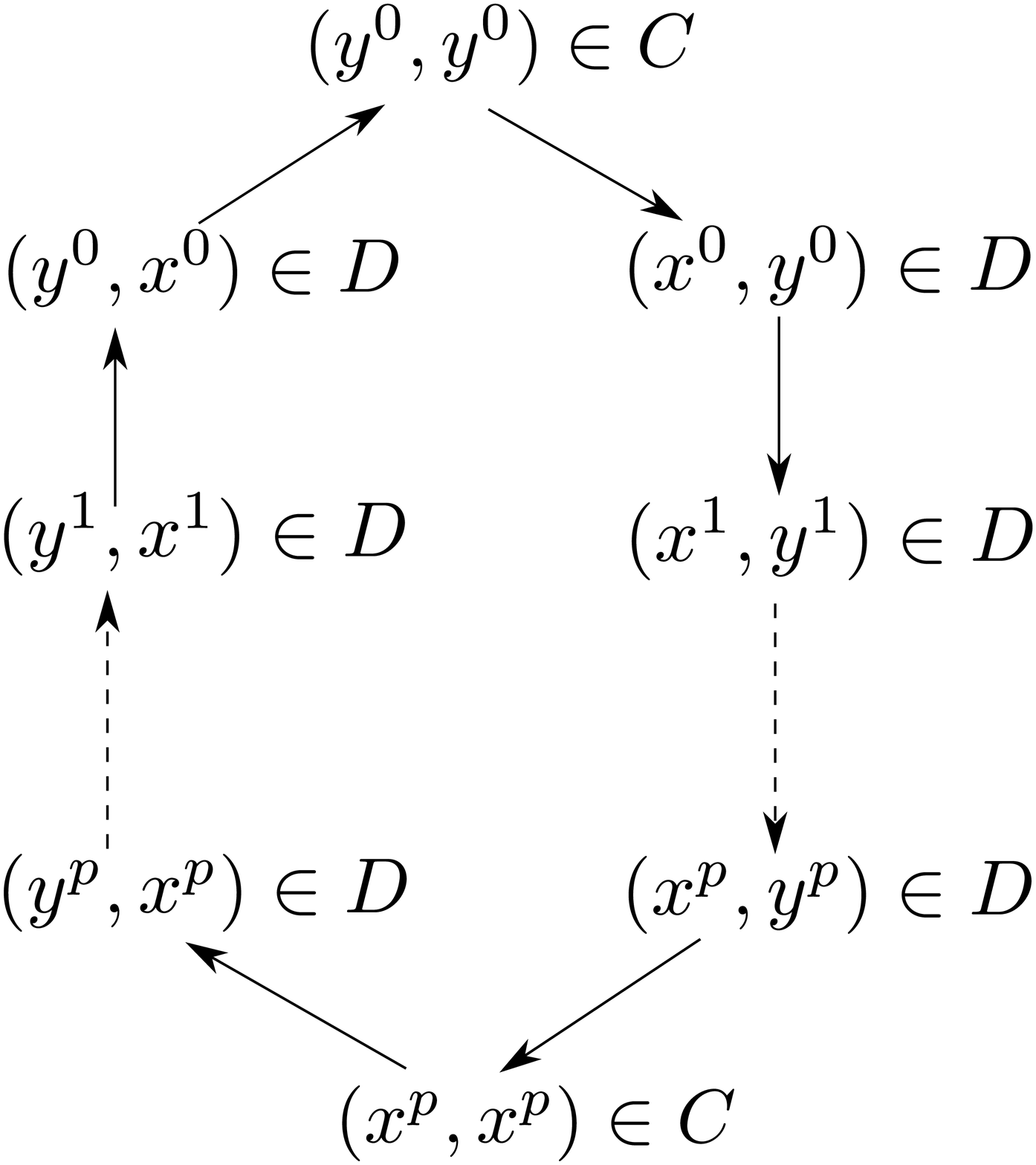}
\vskip 1cm
c)  \includegraphics[width = 4 cm]{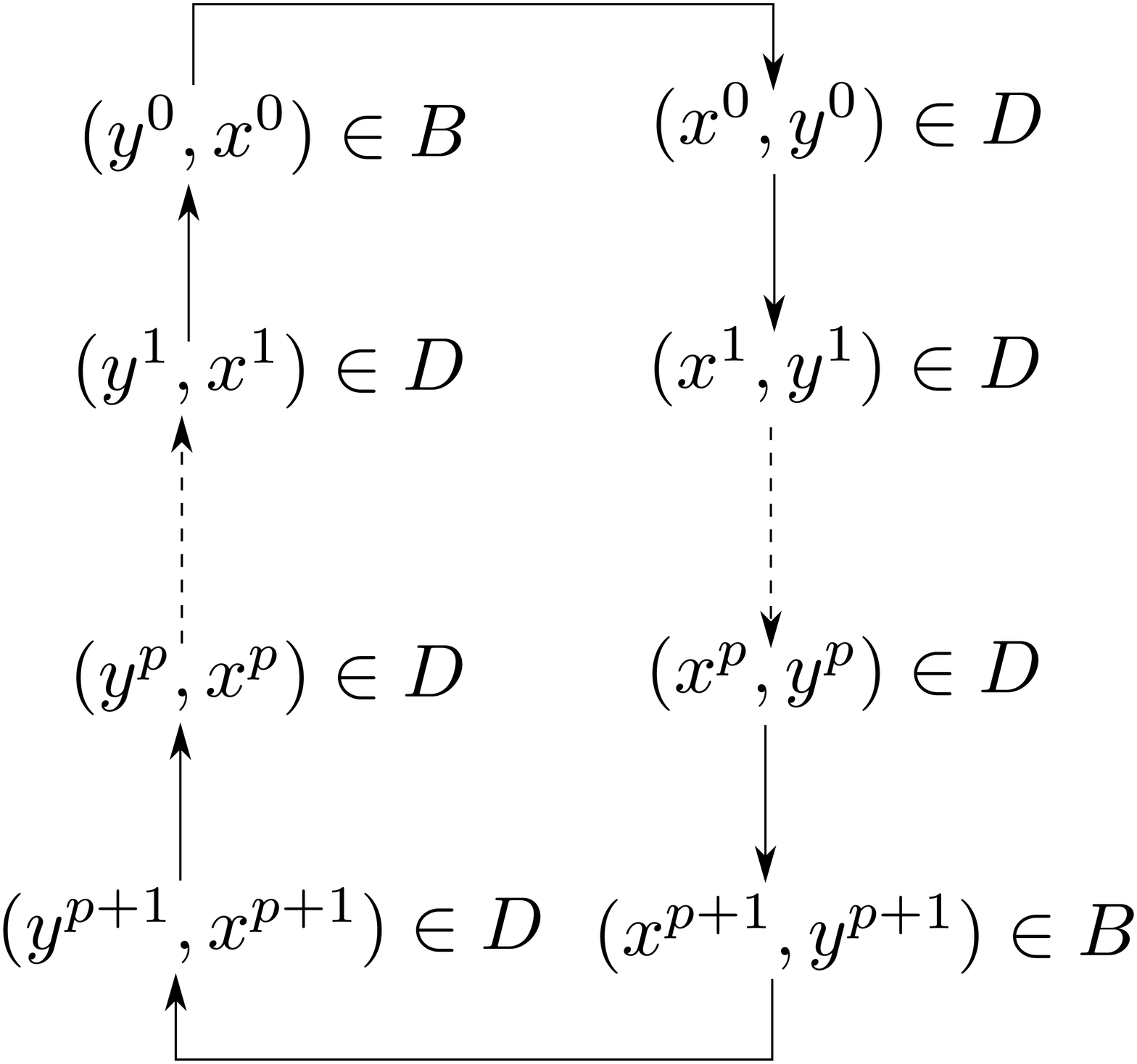} \quad  d) \includegraphics[width = 5 cm]{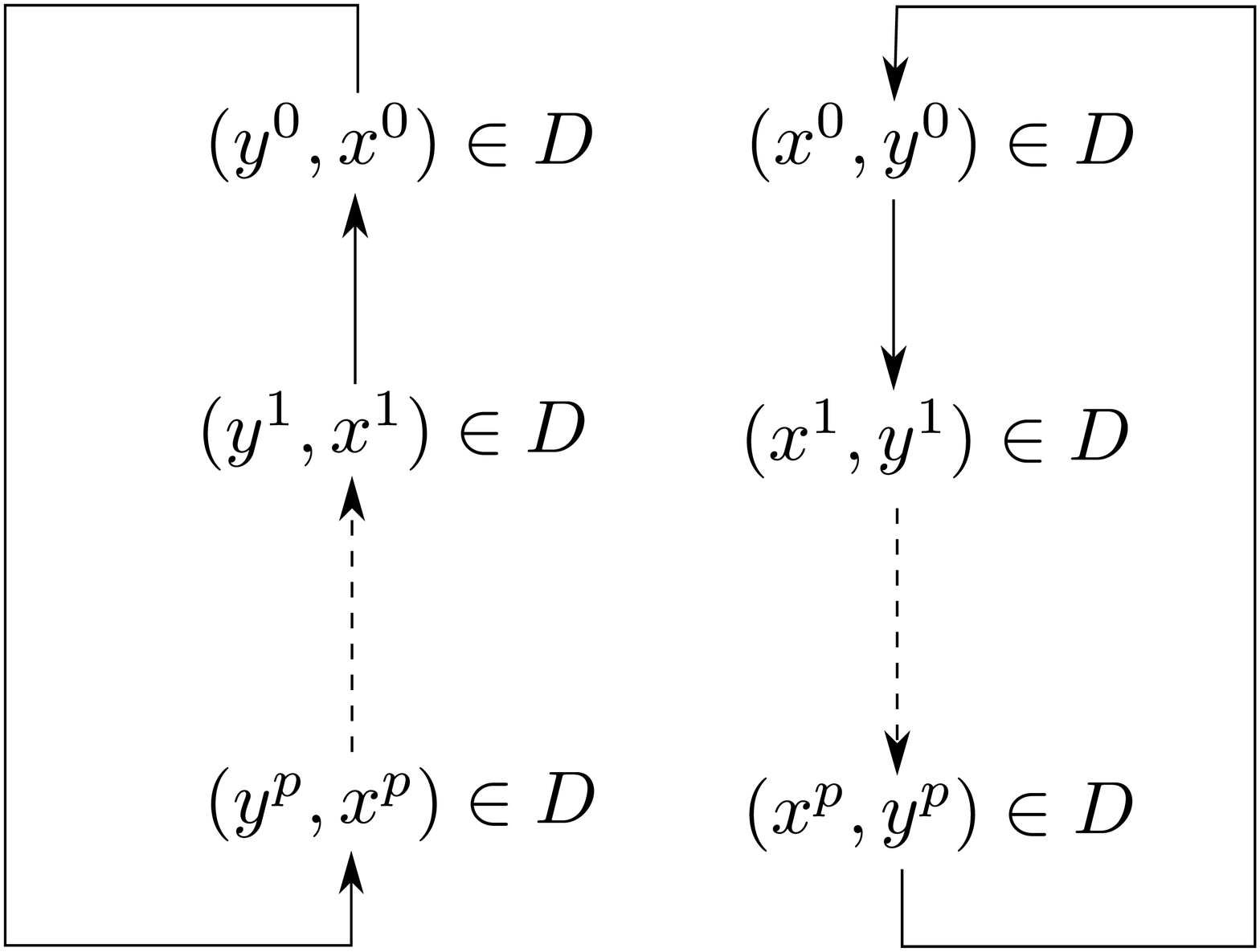}
\end{center}
\caption{ Topology of different type of limit cycles. a) Symmetric limit cycle of type I.   b)  Symmetric limit cycle of type II.  c)  Symmetric limit cycle of type III. d) Two asymmetric limit cycles.}\label{fig_cycletype}
\end{figure}



\subsection{ Attractors classification of Q2R}
 The following (main) Theorem shows that the only possible limit cycles existing in Q2R are the four ones defined above.
\begin{thm}[Attractors classification of Q2R]\label{thm_cycle-clasiff}
Let $\mathcal{C}$ be a limit cycle of Q2R with period $T\in\mathbb{N}$. Then $\mathcal{C}$ is of type 
S-I, S-II, S-III or AS.
\end{thm}
\begin{proof}
Let $\mathcal{C}$ be a limit cycle of Q2R with period $T\in\mathbb{N}$.\\ 
If $T=1$ or $T=2$ then, by definition, $\mathcal{C}$ 
is of type S-I or S-III, respectively.\\
Let $T\geq 3$, { then}, by Corollary \ref{cor_c3props}-(i), 
$\mathcal{C}$ has at least one configuration of the type $D$, $(x^0,y^0)$, which, according with Corollary 
\ref{cor_c3props}-(iii), it may come from a  { configuration of the } type $B$, $C$ or $D$.  Previous statement
allows us to consider only three cases for the limit cycle $\mathcal{C}$:
\begin{itemize}
\item[(C1)] $\mathcal{C}$  has only configurations of type $D$. 
In this case, $\mathcal{C}$ has the form $(x^0,y^0)\to (x^1,y^1)\to \cdots \to (x^p,y^p)\to 
(x^0,y^0)$ with $p=T-1$ as in Figure \ref{fig_cycletype}d)-right, {\it i.e.}, 
$\mathcal{C}$ is of type AS. \footnote{Note that, by Corollary \ref{cor_sec-sim-trans}, 
the symmetric configurations of the previous limit cycle produces a different limit cycle 
$\mathcal{C}'$ of the form $(y^0,x^0)\to (y^p,x^p)\to\cdots \to (y^1,x^1)\to (y^0,x^0)$ also 
composed  by configurations of type $D$, as shown in Figure \ref{fig_cycletype}-d (left), {\it i.e.}, 
$\mathcal{C}'$ is also of type AS.}

\item[(C2)] $\mathcal{C}$ has a configuration $(x^0,y^0)\in D$ coming from a configuration 
$(y^0,x^0)\in B$, as in Figures \ref{fig_cycletype}-a) or \ref{fig_cycletype}c). Because of Corollary \ref{cor_c3props}-(ii), a configuration of type $D$ could (eventually) evolve to another configuration of type $D$, so, we can consider $p\in\mathbb{N}_0$ 
as the maximum time step such that $[(x^0,y^0)\to (x^1,y^1)\to \cdots \to (x^p,y^p)]\in\mathcal{C}$ 
and $(x^t,y^t)\in D$, $\forall i\in\{0,...,p\}$. Hence, we have only two possible evolutions for 
the configuration $(x^p,y^p)\in D$:
\begin{itemize}
\item[(i)] $(x^p,y^p)\to (x^{p+1},y^{p+1})\in B$.\\
In this case, by Corollary \ref{cor_c3props}-(iv), $(x^{p+1},y^{p+1})\to (y^{p+1},x^{p+1})\in D$ 
and $(y^{p+1},x^{p+1})\to(y^p,x^p)\to\cdots \to (y^1,x^1)\to (y^0,x^0)$ (by Corollary \ref{cor_sec-sim-trans}), 
completing $\mathcal{C}$ with an even period $T=2(p+2)$ as shown  in Figure \ref{fig_cycletype}c), i.e., 
$\mathcal{C}$ is of type S-III.

\item[(ii)] $(x^p,y^p)\to (x^{p+1},y^{p+1})\in C$.\\ 
In this case, $(x^{p+1},y^{p+1})=(x^p,x^p)$ and, again,  the Corollary \ref{cor_sec-sim-trans} justifies 
both: $(x^p,x^p)\to (y^p,x^p)$ and $(y^{p},x^{p})\to(y^{p-1},x^{p-1})\to\cdots \to (y^1,x^1)\to (y^0,x^0)$, 
completing $\mathcal{C}$ with an odd period $T=2(p+1)+1$ as shown  in Figure \ref{fig_cycletype}a), i.e., 
$\mathcal{C}$ is of type S-I. 
\end{itemize}
\item[(C3)] $\mathcal{C}$ has a configuration $(x^0,y^0)\in D$ coming from a configuration 
$(y^0,y^0)\in C$. By considering $p\in\mathbb{N}_0$ the same as (C2), again we have only two possible 
evolutions for $(x^p,y^p)\in D$: $(x^{p+1},y^{p+1})\in B$ or $(x^{p+1},y^{p+1})\in C$, 
but the analysis in both cases are similar to (i) and (ii) of the previous case (C2), respectively.  In the case (i), $\mathcal{C}$ 
ends being of type S-I with an odd period $T=2(p+1)+1$ as shown  in Figure \ref{fig_cycletype}a) 
while that, for the case (ii), $\mathcal{C}$ ends up being of type S-II with an even period 
$T=2(p+2)$ as shown in Figure \ref{fig_cycletype}b).
\end{itemize}
\end{proof}

\begin{rem}\label{rem_cyclededuction}
Let $\mathcal{C}$ be a limit cycle of Q2R with period $T\in\mathbb{N}$ and $(x,y)\in\mathcal{C}$. 
Some particular information about $T$ and the configuration $(x,y)$ can help us to deduce the particular type of $\mathcal{C}$.  In fact:


\begin{itemize}
\item If $T$ is odd, then $\mathcal{C}$ could be of type S-I or AS only. If $T$ is even, then $\mathcal{C}$ 
could be of type S-II, S-III or AS.
\item If $(x,y)\in B$ then, if there exists $(x',y')\in\mathcal{C}$ such that $[(x',y')\neq (x,y)]\wedge [(x',y')\in B]$, then $\mathcal{C}$ 
 is of type S-III, otherwise, $\mathcal{C}$ is of type S-I. 
\item If $x=y$ ({\it i.e.}, $(x,y)\in\{A,C\}$), then, if there exists $(x',y')\in\mathcal{C}$ such that $[(x',y')\neq (x,y)]\wedge [x'=y']$, then $\mathcal{C}$ is of type S-II, otherwise,  $\mathcal{C}$ is of type S-I. 
\end{itemize}

\end{rem}

\begin{rem}
The asymmetric limit cycles always appear in pairs (in the sense that all the symmetric configurations 
of one limit cycle belongs into the other limit cycle). Furthermore, we know that $(x,y)\in D$ 
means $[x\neq y]\wedge [\phi (x)\neq\mathbb{1}]$ (regardless the value of $\phi (y)$) but, if $(x,y)$ 
belongs {to} an asymmetric limit cycle, then $(x,y)\in D$ and, necessarily, $\phi (y)\neq\mathbb{1}$, i.e., 
\begin{eqnarray}
[x\neq y]\wedge [\phi (x)\neq\mathbb{1}] \wedge [\phi (y)\neq\mathbb{1}].\label{eq_D-phi}
\end{eqnarray}
The converse relation is not true, i.e., if $(x,y)$ satisfy (\ref{eq_D-phi}) then not necessarily 
$(x,y)$ belongs to an asymmetric limit cycle (see the limit cycle (\ref{cycleSII})  in Appendix \ref{AppSII} as a counterexample).
\end{rem}

\begin{defn}\label{def_nu-ntype}
We denote by $\nu_{\rm SI}(T)$, $\nu_{\rm SII}(T)$, $\nu_{\rm SIII}(T)$ and $\nu_{\rm AS}(T)$ as the number of configurations belonging to a limit cycle of period $T$ and of type S-I, S-II, S-III and AS, respectively. Similarly, $n_{\rm SI}(T)$, $n_{\rm SII}(T)$, $n_{\rm SIII}(T)$ and $n_{\rm AS}(T)$ denote the number of limit cycles of period $T$ and type S-I, S-II, S-III and AS, respectively. Notice that:
\begin{eqnarray*}
 \nu(T) &=& \nu_{\rm SI}(T)+\nu_{\rm SII}(T)+\nu_{\rm SIII}(T)+\nu_{\rm AS}(T)\\
 n(T) &=& n_{\rm SI}(T)+n_{\rm SII}(T)+n_{\rm SIII}(T)+n_{\rm AS}(T)\\
 n_{q}(T) &=& \frac{\nu_{q}(T)}{T}, \quad \text{with } q\in\{{\rm SI},{\rm SII}, {\rm SIII},{\rm AS}\}.
\end{eqnarray*}
\end{defn}

The following result shows that the sets S-I and S-II are naturally rare in the whole phase space.
\begin{thm}\label{teo_rica}
$\displaystyle\sum_{T\geq 1} \left(n_{\rm SI}(T) + 2 n_{\rm SII}(T) \right) = | \Omega^2_{xx}|=2^N$.
\end{thm}
\begin{proof}
Because of the third point of Remark \ref{rem_cyclededuction}, one has that if $(x,x)\in\Omega^2_{xx}$, then, $(x,x)$ belongs only to a cycle of type S-I or S-II. Therefore, this Theorem is true because the limit cycles of type S-I include one configuration in $\Omega^2_{xx}$ while that those cycles of type S-II include two configurations in $\Omega^2_{xx}$.
\end{proof}

\subsection{On the size of attractor's sets: a combinatory approach.}
In this section we show that the number of fixed points, $|P_1|$, will be the fundamental quantity that
determine the size of the different sub-spaces of $\Omega^2$. Our starting point relates the 
neighborhoods of $x\in\Omega$ with those of $-x\in\Omega$.

\begin{prop}\label{prop_x-x}
Let $x\in\Omega$. Then: $\phi (x)=\mathbb{1}\Leftrightarrow \phi (-x)=\mathbb{1}$.
\end{prop}
\begin{proof}
The null neighborhoods of $x\in\Omega$ have two -1's and two 1's.  This is kept for 
$-x\in\Omega$, so: $\phi (x)\neq\mathbb{1} \Leftrightarrow \phi (-x)\neq\mathbb{1}$.
\end{proof}
 
The second result relates fixed points with limit cycles of period 2.
\begin{thm}\label{thm_fp-cl2}
$\{(x,x),(y,y)\}\subseteq P_1 \Leftrightarrow \{(x,y),(y,x)\}\subseteq P_2$ 
\end{thm}
\begin{proof}
\begin{eqnarray}
\{(x,x),(y,y)\}\subseteq P_1 & \Leftrightarrow & [x\neq y]\wedge [\phi (x)=\mathbb{1}]\wedge [\phi (y)=\mathbb{1}] \text{, by Theorem \ref{MMM1}.}\nonumber\\
& \Leftrightarrow & \{(x,y),(y,x)\}\subseteq P_2 \text{, by Theorem \ref{MMM2}.}\nonumber
\end{eqnarray} 
\end{proof}

 As a first consequence of the previous Theorem, we have the following result:
\begin{cor}
Let $x\in\Omega$ such that $\phi (x)=\mathbb{1}$. Then, in the Q2R dynamics: 
\begin{itemize}
\item $(x,x)$ and $(-x,-x)$ are two different fixed points, as well as; 
\item $(x,-x)\to (-x,x)\to (x,-x)$ is a limit cycle of period 2.
\end{itemize}
\end{cor}
\begin{proof}
Let $x\in\Omega$ such that $\phi (x)=\mathbb{1}$, then, by Theorem \ref{MMM1} and Proposition \ref{prop_x-x}, 
$\{(x,x),(-x,-x)\}\subseteq P_1$. {Therefore}, by Theorem \ref{thm_fp-cl2}, $\{(x,-x),(-x,x)\}\subseteq P_2$. Finally, 
 applying the Q2R rule, we have that $(x,-x)\to (-x,x)\to (x,-x)$ is, in fact, a limit cycle of period 2.
\end{proof}

\subsection{Size of the sets $P_1$, $P_2$ and $P_3$.}
 A second consequence of Theorem \ref{thm_fp-cl2} shows that the number of limit cycles of period 2 is larger than the number of fixed points, moreover:

\begin{cor}\label{cor:size-p2}
$|P_2|=|P_1|(|P_1|-1)$. 
\end{cor}
\begin{proof}
\begin{eqnarray*}
     |P_2|&=& 2\cdot\text{(number of limit cycles of period 2), by definition of } P_2\nonumber\\ 
     &=& 2\cdot{|P_1| \choose 2}\text{, by Theorem \ref{thm_fp-cl2}}\nonumber\\
     &=& |P_1|(|P_1|-1).
\end{eqnarray*} 
\end{proof}

 Moreover, because of Remark \ref{rem_p1}, $P_1=A=\left\{ (x,y)\in\Omega_{xx}^2  / 
 \phi(x)=\mathbb{1} \right\}$. 
In other words, $P_1=\{(x,x)\in\Omega^2:\, \phi(x)=\mathbb{1}\}$. Thus, 
$|P_1|=|\Phi_{\mathbb{1}}|$ where $\Phi_{\mathbb{1}} = \left\{ x\in \Omega / \phi(x)=\mathbb{1}\right\}$, {\it i.e.}, $\Phi_{\mathbb{1}}$ is the set of states without null neighborhood. 

\begin{defn}\label{Def:BW} 

Let $x\in\Omega$. We define the \textit{staggered-state} \cite{takesue2} $x_B$ as the state $x$ restricted to the nodes ${\bm k}\in \{1,...,N\}$ such that
$k_1+k_2$ is even. Analogously is defined the staggered-state $x_W$ but for the nodes ${\bm k}\in \{1,...,N\}$
such that 
$k_1+k_2$ is odd. In other words, $x_B$ and $x_W$ can be seen 
as the black and white fields in the ``chessboard $x\in\Omega$'' respectively, and, such that its superposition 
-- denoted with the $\uplus$ symbol -- is $x$, i.e.:
$$x = [x_B\uplus x_W ] \in \Omega.$$
We will use the subindices $(\cdot)_B$ and $(\cdot)_W$ to refer to the corresponding restriction 
of the element $(\cdot)$ which we are working in order that:
$$(\cdot) = [(\cdot)_B\uplus (\cdot)_W] \in \Omega.$$

In particular, we define the \textit{chessboard} states $\mathbb{1}_{BW}$ and $\mathbb{1}_{WB}$ of $\Omega$ as follows:
$$\mathbb{1}_{BW}\equiv[-\mathbb{1}_B]\uplus \mathbb{1}_W \quad\text{ and }\quad \mathbb{1}_{WB}\equiv -\mathbb{1}_{BW}=\mathbb{1}_{B}\uplus [-\mathbb{1}_W].$$
\end{defn}

\begin{rem}\label{rem:defVonNeuman}
 This construction using the staggered-states is particularly useful in the case of the von Neuman neighborhood. Moreover, in the particular case of $L$ being an odd number, the staggered-states loss its utility (see the first fact of Remark \ref{rem_facts}).
\end{rem}

\begin{rem}\label{rem:defBW}
According with the above definition, notice the following statements:
\begin{enumerate}
 \item \label{f1} $\Omega=\Omega_B\uplus \Omega_W$. 
 \item \label{f2} $|\Omega|=2^N \Rightarrow |\Omega_B|=|\Omega_W|=2^{N/2}$.
 \item \label{f3} $x\in\Omega \Leftrightarrow [x_B\in\Omega_B]\wedge [x_W\in\Omega_W]$.
 \item \label{f4} $\phi (\mathbb{1}_{BW})=\phi (\mathbb{1}_{WB})=\mathbb{1}$.
\end{enumerate}
\end{rem}
  The following two propositions establish conditions for the existence of fixed points and limit cycles of period 2 and 3.

\begin{prop}\label{prop:p1-p3always} 
$\forall L$ even, $4\leq |P_1| < |P_2|$.  
\end{prop}
\begin{proof}
Because of definition \ref{rem:phi1-1} and remark \ref{rem:defBW}: $\phi(-\mathbb{1})=\phi(\mathbb{1})=\phi (\mathbb{1}_{BW})=\phi (\mathbb{1}_{WB})=\mathbb{1}$.  Hence,
by Theorem \ref{MMM1}: 
$$\left\{(\mathbb{1},\mathbb{1}),(-\mathbb{1},-\mathbb{1}),(\mathbb{1}_{BW},\mathbb{1}_{BW}),(\mathbb{1}_{WB},\mathbb{1}_{WB})\right\}\subseteq P_1,$$
therefore, $4\leq |P_1| < |P_1|(|P_1|-1)=|P_2|$ (where the last equality is by Corollary \ref{cor:size-p2}).
\end{proof}

\begin{prop}\label{prop:p3always}
$\forall L\geq 4$,  $|P_3|\geq 6N= 6 L^2$.  
\end{prop}
\begin{proof}
 Let $L\geq 4$, and consider an state $x\in\Omega$ of Remark \ref{rem_p2} where we have proved that
$\{(\mathbb{1},x),(x,\mathbb{1}),(x,x)\}\subseteq P_3$, {\it  i.e.}, $|P_3|\geq 3$.

Because the block of -1s may be situated at any point of $x$, there are $N$ equivalent configurations, moreover, because the inverse configuration $(-\mathbb{1},-x)\in P_3$, there are at least $2 N$ limit cycles of period 3. Therefore, 
$|P_3|\geq 6N$.
\end{proof}

\begin{defn}
We denote by $B_{\mathbb{1}}$ and $W_{\mathbb{1}}$ as the sets of staggered-states $x_B\in\Omega_B$ and $x_W\in\Omega_W$ 
without null neighborhoods, respectively. That is: 
\begin{eqnarray}
B_{\mathbb{1}} &=\{u\in\Omega_B:\exists x\in\Omega, [x_B=u] \wedge [\phi(x)_W=\mathbb{1}_W]\}\nonumber\\
W_{\mathbb{1}} &=\{v\in\Omega_W:\exists x\in\Omega, [x_W=v] \wedge [\phi(x)_B=\mathbb{1}_B]\}\nonumber
\end{eqnarray}
\end{defn}

\begin{ex}
Consider the following state $x\in\{-1,1\}^{16}$ and its corresponding state $\phi (x)\in\{-1,1\}^{16}$:
\begin{eqnarray}
x= \left[ \begin{array}{cccc} 
 \fbox{1} & 1 & \fbox{1} & 1 \\
1 & \fbox{-1} & 1 & \fbox{1} \\
\fbox{1} & 1 &\fbox{-1} & 1 \\
1 & \fbox{1} & 1 & \fbox{1} \\
\end{array} \right] \text{ and }  
\phi (x) = \left[ \begin{array}{cccc} 
 1 & \fbox{1} & 1 & \fbox{1}\\
\fbox{1} & 1 & \fbox{-1} & 1\\
1 & \fbox{-1} & 1 & \fbox{1}\\
\fbox{1} & 1 & \fbox{1} & 1\\
\end{array} \right]\nonumber
\end{eqnarray}
The values in the boxes of $x$ {correspond to} the staggered-state $x_B\in\{-1,1\}^8$ while the other values that are not in boxes  correspond to
 $x_W\in\{-1,1\}^8$.  Similarly, the values in the boxes of $\phi(x)$ correspond to $\phi(x)_W$ 
and were obtained with the values of $x_B$ while the other values, that are not in the boxes of $\phi(x)$, correspond to 
$\phi(x)_B$ and were obtained with the values of $x_W$.  
\end{ex}

\begin{rem}\label{rem_facts}
The previous example allows us to understand the following facts, that are direct consequences of the above 
definitions: 
\begin{enumerate}
 \item \label{f1bis} The neighborhoods in $x_B$ are independent of those in $x_W$, when $L$ is even.
 \item \label{f2bis}  $\phi(x)=\mathbb{1}\Leftrightarrow [\phi(x)_B=\mathbb{1}_B]\wedge [\phi(x)_W=\mathbb{1}_W]$. 
In other words, studying the set $\Phi_{\mathbb{1}}$  is equivalent to study the sets $B_{\mathbb{1}}$ 
and $W_{\mathbb{1}}$.
 \item \label{f3bis}  $B_{\mathbb{1}}\uplus W_{\mathbb{1}}=\Phi_{\mathbb{1}}$.
 \item \label{f4bis} $|B_{\mathbb{1}}|=|W_{\mathbb{1}}| $.
\end{enumerate} 
\end{rem}

\begin{defn}
We denote $\beta\equiv |B_{\mathbb{1}}|=|W_{\mathbb{1}}|$. 
\end{defn}

The following corollary will be useful to prove, in Theorem \ref{teo:size-p1-p2}, that  $\beta$ is an even number.

\begin{cor}\label{cor_betapar}
Let $x\in\Omega$. Then:
\begin{enumerate}
\item \label{f1ter} $x_B\in B_{\mathbb{1}}\Leftrightarrow -x_B\in B_{\mathbb{1}}$
\item \label{f2ter} $x_W\in W_{\mathbb{1}}\Leftrightarrow -x_W\in W_{\mathbb{1}}$
\end{enumerate}
\end{cor}
\begin{proof}
Let $x\in\Omega$.
\begin{eqnarray}
     (1)\qquad\quad x_B\in B_{\mathbb{1}} & \Leftrightarrow & \phi(x)_W=\mathbb{1}_W\nonumber\\ 
     & \Leftrightarrow & \phi(-x)_W=\mathbb{1}_W\nonumber\\
     & & \text{(by Proposition \ref{prop_x-x} and fact (\ref{f2bis}) of Remark \ref{rem_facts})}\nonumber\\
     & \Leftrightarrow & -x_B\in B_{\mathbb{1}}\nonumber
    \end{eqnarray}
(2) Is analogous to (1).
\end{proof}

The following Theorem shows that the number of fixed points of Q2R is always an square number 
and a multiple of 4.

\begin{thm}\label{teo:size-p1-p2}
The following statements are true:
\begin{enumerate} 
 \item $\beta=2k$, for some $k\in\mathbb{N}$.
 \item $|P_1|=|\Phi_{\mathbb{1}}|=\beta^2=4k^2$.
 \item $|P_2|=4k^2(4k^2-1)$.
\end{enumerate}
\end{thm}

\begin{proof}
$ $
\begin{enumerate} 
\item $\beta\geq 2$ because always $\{\mathbb{1}_B,-\mathbb{1}_B\}\subseteq B_{\mathbb{1}}$ and 
$\{\mathbb{1}_W,-\mathbb{1}_W\}\subseteq W_{\mathbb{1}}$. The fact that $\beta$ it is even 
is direct from Corollary \ref{cor_betapar}.
 \item As already said:\\ 
$|P_1|=|\{(x,x)\in\Omega^2:\, \phi(x)=\mathbb{1}\}|=|\{x\in\Omega  / \phi(x)=\mathbb{1}\}|=|\Phi_{\mathbb{1}}|$. Since $\beta=2k$ and because of
the facts (\ref{f3bis}) and  (\ref{f4bis}) of Remark \ref{rem_facts}:
$|\Phi_{\mathbb{1}}|=|B_{\mathbb{1}}|\cdot |W_{\mathbb{1}}|=\beta^2=4k^2$. Therefore: $|P_1|=|\Phi_{\mathbb{1}}|=\beta^2=4k^2$, for some $k\in\mathbb{N}$.
 \item   This proof is a direct from Corollary \ref{cor:size-p2} and statement (2) previously proved.
\end{enumerate} 
\end{proof}

\section{A general overview of {the phase space}}\label{SecOmega2} 

 Before concluding, we will explain the Figure \ref{fig_omega} that
 summarizes briefly the phase space partitions in limit cycles accordingly with the main results obtained in the previous sections.
 
\begin{itemize}
\item The whole figure represents $\Omega^2$ that is partitioned according to definition \ref{defOmegaxx}: $\Omega^2=\Omega_{xx}^{2}\cup\Omega_{xy}^{2}$.
\item $\Omega_{xx}^{2}$ ({colored by }yellow and green) is partitioned in ${ A=}P_1$ (yellow) and ${ C=}\Omega_{xx}^{2}-P_1$ (green).
\item $\Omega_{xy}^{2}$ ({colored by } orange and blue) is partitioned in $P_2$ (orange) and $\Omega_{xy}^{2}-P_2$ (blue).
\item The figure {tries to reflect, though not in the real scale,} that {$|\Omega_{xx}^{2}|\ll |\Omega_{xy}^{2}|$} and {$|P_1|\ll|P_2|$}.
\item $P_1\cup P_2$ (yellow and orange) represents the set of configurations without null neighborhoods. For the 
{complementary} configurations (those of green and blue regions), at least one of its states has a null neighborhood.
\item The dashed line,{ that limits} the left region of $\Omega^2$ with colors orange, yellow, green and a part of 
the blue, corresponds to configurations {that are} symmetric limit cycles ({\it i.e.} , limit cycles of type S-I, S-II 
or S-III only). The remaining region of $\Omega^2$ (only in blue) correspond to configurations belonging to asymmetric 
limit cycles (i.e., of type AS).
\item Configurations belonging in $C=\Omega_{xx}^{2}-P_1$ (green) are in limit cycles of type S-I or S-II only. Those 
of type S-I are represented only with one configuration in the green region (the remaining configurations {lying} in the blue region {limited} by the dashed line). Those of type S-II are represented only with two configurations in the green 
region (the remaining configurations {lying} in the blue region enclosed by the dashed line).
\item The limit cycles {with all its configurations belonging to} the leftmost blue region {limited} by the dashed line are exclusively of type S-III.
\item While $P_1$ and $P_2$ are exactly the yellow and orange regions, respectively, the other sets $P_j$ ($j\geq 3$) 
have configurations in the green or blue regions (but not in the yellow nor the orange) as exemplified in the 
figure. 
\end{itemize}

\begin{figure}
\begin{center}
\includegraphics[width = 8 cm]{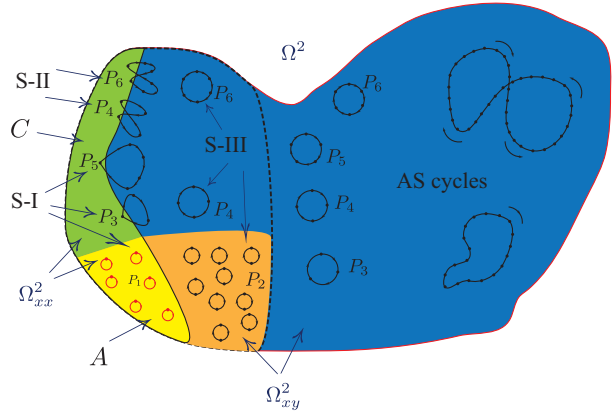} 
\end{center}
\caption{The phase space $\Omega^2$ decomposed according to the results of the previous sections.}
\label{fig_omega}
\end{figure}

The following Table \ref{tab:sizes} complements Figure \ref{fig_omega} by showing the sizes of the different regions of $\Omega^2$ above mentioned  for small system ($N= 16, \, 32,$ and $ 64$).

\begin{table}[h]
\scriptsize
\begin{tabular}{|c|c|c|c|c|c|}
\hline
label & variable & size & $L=4$ &$L=6$ & $L=8$\\ \hline\hline
[b] & $N$ & $L^2$ & 16 & 36 & 64\\ \hline
[c] &  $|\Omega^2|$ & $2^{2N}$ & $2^{32}\sim 4\cdot 10^{9}$& $2^{72}\sim 4\cdot 10^{21}$ & $2^{128}\sim 3\cdot 10^{38}$\\ \hline
[d] & $|\Omega_{xx}^{2}|$ & $2^N$ & $2^{16}=65536$ & $2^{36}\sim 7\cdot 10^{10}$ & $2^{64}\sim 2\cdot 10^{19}$\\ \hline
[e] & $|\Omega_{xy}^{2}|$ & [c]-[d] & $\sim 4\cdot 10^{9}<$[c] & $\sim 4\cdot 10^{21}<$[c] & $\sim 3\cdot 10^{38}<$[c]\\ \hline
[f] & |Yellow| & $|P_1|=\beta^2$ & $34^2$& $584^2$ & $39426^2$\\ \hline
[g] &| Orange| &$|P_2| $= [f]([f]-1) & 1335180 & $\sim 10^{9}$ & $\sim 2\cdot 10^{18}$\\ \hline
[h] & |Green| & [d]-[f]  & 64380 & $\sim 7\cdot 10^{10}<$[d] & $\sim 2\cdot 10^{19}<$[d]\\ \hline
[i] &| Blue| & [c]-([f]+[g]+[h]) & $\sim 4\cdot 10^{9}<$[e] & $\sim 4\cdot 10^{21}<$[e] & $\sim 3\cdot 10^{38}<$[e]\\ \hline
\end{tabular}
\caption{
Summary of the sizes of the main regions discussed in the Figure \ref{fig_omega}. The 1st column labels the values of the ``size'' column. 
The 2nd column has the variables and the main regions explained in Figure \ref{fig_omega}. In the 3rd column are the size formulas for each ``variable'' of the 2nd column.
In the 4th, 5th and 6th columns are the calculations done for $L\in\{4,6,8\}$, respectively.  }
\label{tab:sizes} 
\end{table}

Since the sizes of the different partitions shown in Figure \ref{fig_omega} essentially depends on $\beta$, this value 
was computationally calculated in Table \ref{tab:sizes} by generating all the staggered-states $x_B\in\Omega_B$ 
without null neighborhoods ({\it i.e.}, $\phi(x)_W = \mathbb{1}$) in order to construct the set $B_{\mathbb{1}}$. 
This procedure gave us a computable size of $|\Omega_B|=2^{N/2}$, for $N= 16, \, 32,$ and $ 64$.

\section{Discussion}\label{Discussion}

Because of the relevance of an accurate knowledge of the phase space in complex dynamical systems with many degrees of freedom, we attempted a classification of the phase space of the Q2R cellular automaton which is in close connection with the Ising model and its statistical properties. The Q2R model is reversible and essentially all results of the present paper follow after the Lemma \ref{LemRev} (on Reversibility). The main results in the present paper are: Theorem \ref{thm_cycle-clasiff} that shows a fully classification of all Q2R attractors in four types of limit cycles consisting of symmetric and asymmetric ones. Moreover,  a general overview of the phase space decomposition has been provided.  
Besides, some specific results for small period limit cycles are the following: Theorem \ref{teo:size-p1-p2} that shows both,  that the total number of fixed points is of the form $|P_1|=\beta^2 =4 k^2$, with $k\in\mathbb{N}$ and that the number of configurations belonging in a period-2 limit cycle is  $ |P_2|=\beta^2 (\beta^2-1)$. Finally, Theorems \ref{MMM1}, \ref{MMM2} and \ref{teo_cl3} as well as Propositions \ref{prop:p1-p3always} and \ref{prop:p3always} characterize the attractors with period lower or equal to 3 and guarantees that almost always will be present in the Q2R dynamics.

We end this paper with the following open remarks:
\begin{enumerate}
\item In the context of statistical mechanics it is important to well characterize the phase space restricted to a subset of constant energy. Regarding the above, preliminary numerical simulations show that limit cycles of different periods can have exactly the same level of energy although we have not yet been able to detect common patterns that allow us to characterize them, so, as a future work, we plan to refine the above numerics by using the limit cycle classification presented in this work in order to detect such a patterns.

\item In this new framework, it is evident that the presence of null neighborhoods is determinant for the presence of limit cycles of period 3 or higher and the absence of null neighborhoods, results in fixed points or period-2 limit cycles. In this context, variations of the Q2R model maybe considered for other regular lattices, however, if the neighborhood possesses an odd number of nodes the further dynamics is trivial. More formally,

\begin{thm}
If the neighborhood considered for the definition of Q2R is composed by an odd number of nodes, then the Q2R 
dynamic has only fixed points and period-2 limit cycles.
\end{thm}
\begin{proof}
If the neighborhood has an odd number of nodes, then, given $x\in\Omega$, all its neighborhoods 
will have more 1s than -1s or vice-versa. In both cases, its sum will never be 0. In other words,
$\phi (x)=\mathbb{1}$, $\forall x\in\Omega$. Thus, by Theorems \ref{MMM1} and \ref{MMM2}, the 
Q2R dynamic will have only fixed points and period-2 limit cycles.
\end{proof}

\item Notice that we do not have a mathematical expression for $\beta$.  In this context, all numerical values involving $\beta$ in this paper, such as those of Table \ref{tab:sizes}, were obtained checking state by state if they have or not a null-neighborhood.

\item At a computational level, it is important to note that, to make a correct dynamic study of Q2R, initial 
configurations of both $\Omega_{xx}^{2}$ and $\Omega_{xy}^{2}$ must be considered because if only the first one is considered, 
then it will never be possible to obtain, for instance, period-2 limit cycles nor asymmetric limit cycles. On the other hand, 
if only $\Omega_{xy}^{2}$ is considered, then it will never be possible to obtain fixed points.

\item For the dynamical characterization of limit cycles of period-4 and higher, we have not proven general conditions of existence neither we are not able to compute its cardinality, however, a period $T$ limit cycle is characterized by 
\begin{eqnarray}
 (x^0,y^0) \to (x^1,y^1)  &  \to & (x^2,y^2)\to \cdots \nonumber\\
& \cdots &  \to (x^{T-1},y^{T-1}) \to (x^T,y^T)=  (x^0,y^0),\nonumber\end{eqnarray}
with $x^1=  y^0\odot \phi(x^0)$, $y^1=x^0$, $ x^2=y^1 \odot \phi(x^1)$, \dots \, , $x^T =   y^{T- 1}\odot \phi(x^{T-1})=x^0$, $y^T =   x^{T-1}=y^0$.
Therefore, one concludes the following necessary (but not sufficient) conditions:
 

\begin{enumerate}
\item For $T$ even, we have two, {\it a priori}, independent conditions:

\begin{eqnarray}
\phi(x^0) \odot\phi(x^2) \odot\phi(x^4) \odot \cdots \odot \phi(x^{T-4}) \odot \phi(x^{T-2})  & = & \mathbb{1}  \nonumber \\
\phi(x^1) \odot\phi(x^3) \odot\phi(x^5) \odot \cdots \odot \phi(x^{T-3}) \odot \phi(x^{T-1})  & = & \mathbb{1}   
.  \label{CycleEvenCondition}
\end{eqnarray}
\item For an odd period, one has the necessary condition
\begin{eqnarray}
\phi(x^0) \odot\phi(x^1) \odot\phi(x^2) \odot \cdots \odot \phi(x^{T-2}) \odot \phi(x^{T-1})  & = & \mathbb{1}   
.  \label{CycleOddCondition}
\end{eqnarray}
\end{enumerate}

As an example the period four limit cycles posses three different types of limit cycles, the type S-II, the type S-III and the  type AS (see Figure \ref{Period4}). 
 \begin{figure}[htbp!]
\begin{center}
a) \includegraphics[width=2.5cm]{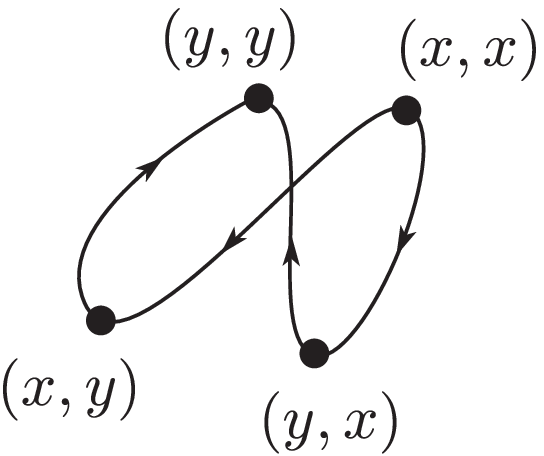}  \, b) \includegraphics[width=3cm]{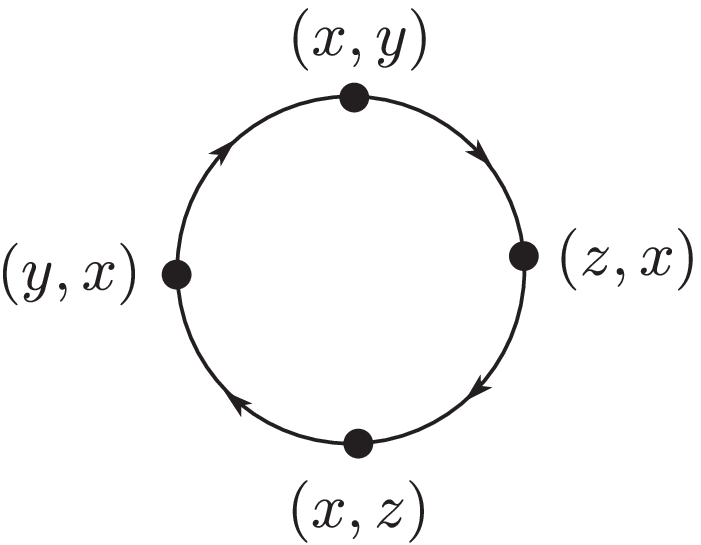}  \, 
c) \includegraphics[width=3cm]{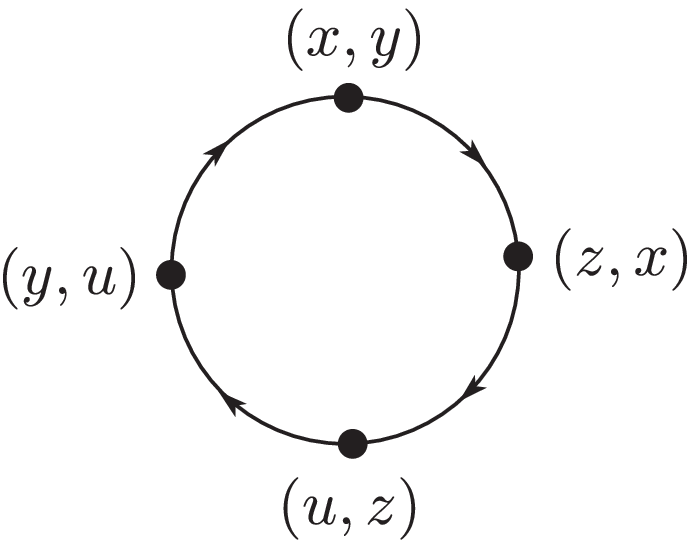}
\end{center}
  \caption{\label{Period4} a) A period-4 limit cycle $(x,x) \to (y,x)   \to (y,y)  \to (y,x)  \to (x,x)$ (type S-II).  
  b) A period-4 limit cycle $(x,y) \to (z,x)  \to (x,z)  \to (y,x)  \to (x,y)$ (type S-III). c) A period-4 limit cycle $(x,y) \to (z,x)  \to (u,z)  \to (y,u)  \to (x,y)$ (type AS).}
\end{figure}

The period five limit cycles posses two different types of limit cycles, the type S-I and the type AS (see Figure \ref{Period5}). 
 \begin{figure}[htbp!]
\begin{center}
a) \includegraphics[width=3.cm]{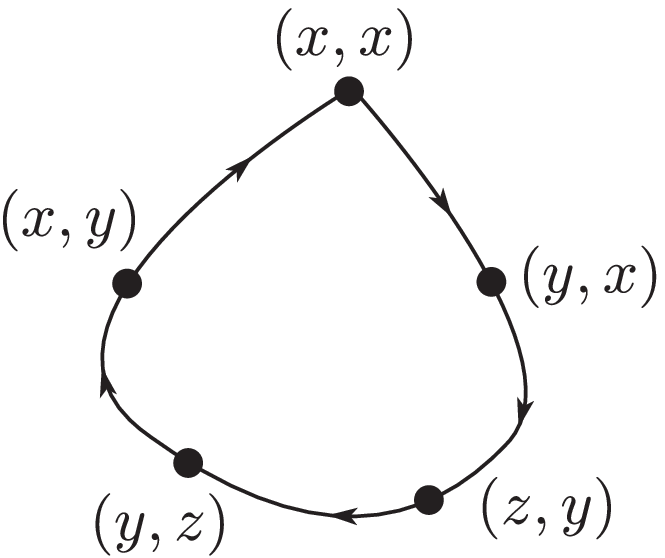}  \, b) \includegraphics[width=3.cm]{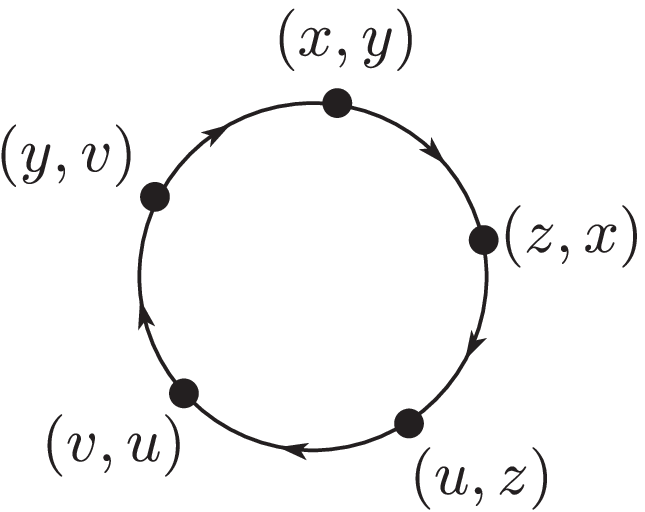} 
\end{center}
  \caption{\label{Period5} (a) A period 5 limit cycle 
  $(x,x) \to (y,x)   \to (z,y)  \to (y,z)   \to (x,y)   \to (x,x) $.  (b) Example of a limit cycle  
  $(x,y) \to (z,x)  \to (u,z)  \to (v,u)  \to (y,v)  \to (x,y)  $ this limit cycle does not exist 
  in the case of a $4\times 4$ lattice.  }
\end{figure}

\item Because of the special topology of the limit cycles of type S-I and S-II (see Figure \ref{fig_cycletype}-a \& \ref{fig_cycletype}-b) these limit cycles are fully characterized by a simpler set of conditions.

 Let be the sequence:
\begin{eqnarray}
x^1 &=& x^0\odot   \phi(x^0)  \quad \text{for the first step} \nonumber\\
x^{t+1} &=& x^{t-1}\odot   \phi(x^t)   \quad \text{ for }\, t=\{1,2,3,\dots, T-1\} ,\end{eqnarray}  
then,  the closing conditions for limit cycle  for an even periodic limit cycle (\ref{CycleEvenCondition}) imposes 
\begin{eqnarray}
\phi(x^0) \odot \phi(x^1) \odot \phi(x^2) \odot\phi(x^3) \odot \cdots \odot \phi(x^{T-4}) \odot \phi(x^{T-1})  & = & \mathbb{1}, \label{CycleEvenConditionSII}
\end{eqnarray}
while, for an odd period,  the condition (\ref{CycleOddCondition}) simplifies to
\begin{eqnarray}
\phi(x^{(T-1)/2})  & = & \mathbb{1}   
.  \label{CycleOddConditionSI}
\end{eqnarray}
Because the even and odd cases follow quite different conditions we conjecture that:

\begin{conj} \label{Conj:Odd}  Let be odd period $T = 2 q +1 $ with $q \in \mathbb{N}   $, and let   $x^1 = x^0\odot   \phi(x^0)  $, together with $x^{t+1} = x^{t-1}\odot   \phi(x^t)   $ for 
 $t\in \{1,2,3,\dots, q -1\} $, then, the pair $(x^0,x^0)$ belongs to a type S-I limit cycle of period $T=2 q +1 $, iff
                \begin{eqnarray} 
                \phi(x^0) \neq  \mathbb{1}  \wedge  \phi(x^1)   \neq  \mathbb{1}    \wedge  \dots   \wedge  \phi(x^{q-1})   \neq  \mathbb{1}     \wedge  
                  \phi(x^{q})  &=&  \mathbb{1} .\nonumber
                 \end{eqnarray}  
\end{conj}
This condition appears to be an easy way to compute the cardinality of the sets S-I with odd periods.

\begin{conj} \label{Conj:Even}  Let be the even period  $T = 2 p$, with $p$ a prime number and let $x^1 = x^0\odot   \phi(x^0)  $, and $x^{t+1} = x^{t-1}\odot   \phi(x^t)   $ for $t=\{1,2,3,\dots, p-2\} $, then the pair $(x^0,x^0)$ belongs to a type S-II limit cycle of a period $T=2 p$ iff 
                \begin{eqnarray} 
                \phi(x^0) \neq  \mathbb{1}  \wedge  \phi(x^1)   \neq  \mathbb{1}    \wedge  \dots   \wedge  \phi(x^{p-2})   \neq  \mathbb{1}     \wedge   & & \nonumber\\
                 \phi(x^0) \odot   \phi(x^1) \odot    \dots   \odot   \phi(x^{p-1})  &=&  \mathbb{1} .\label{eeqn:ConjOdd}
                 \end{eqnarray}  
\end{conj}
As before, the cardinality of the sets S-II with even periods  $T = 2 p$ ($p$ a prime number), nevertheless, the general even period case requires more careful considerations. Essentially, there is a double counting, {\it e.g.} the case of period $T= 2\times 4 =8$ the condition (\ref{eeqn:ConjOdd})  counts simultaneously the  period 8 and period 4 configurations.

\end{enumerate}

\section{Appendix: Exact results for the case $4\times 4$.}\label{AppEx}

\subsection{Decomposition of $\Omega^2$}

The $4\times 4$ lattice is the largest possible phase space, with $2^{32}$ configurations, that can be fully scanned {numerically}. 
The next $6\times 6$ lattice possesses $2^{76}$ configurations, making impossible this task.

In Tables \ref{TableConfigs} and \ref{TableCycles}, we provide the exact distributions for the number of configurations and for the number of limit cycles  accordingly with its period and the cycle type, respectively. One observes that the number of odd period limit cycles are rare. As a general rule, the number of configurations (and the number of limit cycles) of even period limit cycles are much larger than the odd ones limit cycles. We also notices that 
the periods 7, 11, 13, etc. are missing. Moreover, the largest odd period is 27.  The reason why some periods exist and other does not is still an open problem.

\begin{table}[h!]
\scriptsize
\begin{tabular}{|c||c|c|c||c||c|}
\hline
$T$ & { $\nu_{\rm SI}(T)$} & { $\nu_{\rm SII}(T)$} & { $\nu_{\rm SIII}(T)$} & { $\nu_{\rm AS}(T)$} & $|P_T|$\\ \hline\hline
1 & 1,156 & 0 & 0 & 0 & 1,156\\ \hline
2 & 0 & 0 & 1,335,180 & 0 & 1,335,180\\ \hline
3 & 4,128 & 0 & 0 & 768 & 4,896\\ \hline
4 & 0 & 14,456 & 20,556,256 & 48,384,408 & 68,955,120\\ \hline
5 & 1,920 & 0 & 0 & 0 & 1,920\\ \hline
6 & 0 & 10,560 & 15,219,936 & 20,054,976 & 35,285,472\\ \hline
8 & 0 & 42,752 & 58,399,744 & 235,007,232 & 293,449,728\\ \hline
9 & 3,456 & 0 & 0 & 4,608 & 8,064\\ \hline
10 & 0 & 7,680 &5,174,400  & 2,941,440 & 8,123,520\\ \hline
12 & 0  & 99,648 &132,294,144  & 655,316,928 & 787,710,720\\ \hline
18 & 0 & 69,120 &18,824,832  & 143,732,736 & 162,626,688\\ \hline
20 & 0 &19,200  &17,295,360  & 53,694,720 & 71,009,280\\ \hline 
24 & 0 & 27,648 &115,703,808  & 536,220,672 & 651,952,128\\ \hline
27 & 0 & 0 & 0 &6,912  & 6,912\\ \hline
30 & 0 & 0 & 15,851,520 & 2,949,120 & 18,800,640\\ \hline
36 & 0 & 0 & 51,038,208 & 333,388,800 & 384,427,008\\ \hline
40 & 0 & 76,800 &26,296,320  &246,420,480  & 272,793,600\\ \hline
54 & 0 & 186,624 & 0 & 242,721,792 & 242,908,416\\ \hline
60 & 0 & 0 & 33,177,600 & 113,172,480 & 146,350,080\\ \hline
72 & 0 & 0 & 47,333,376 & 162,201,600 & 209,534,976\\ \hline
90 & 0 & 0 & 13,271,040 &17,694,720  & 30,965,760\\ \hline
108 & 0 & 0 & 0 & 329,508,864 & 329,508,864\\ \hline
120 & 0 & 0 & 0 &200,540,160  & 200,540,160\\ \hline
180 & 0 & 0 & 0 &30,965,760  & 30,965,760\\ \hline
216 & 0 & 0 & 0 &179,601,408  & 179,601,408\\ \hline
270 & 0 & 0 & 0 & 26,542,080 & 26,542,080\\ \hline
360 & 0 & 0 & 0 & 61,931,520 & 61,931,520\\ \hline
540 & 0 & 0 & 0 & 26,542,080 & 26,542,080\\ \hline
1080 & 0 & 0 & 0 & 53,084,160 & 53,084,160\\ \hline\hline
Total & 10,660 & 554,488 &  571,771,724 & 3,722,630,424 & $2^{32}$\\ \hline
\end{tabular}
\caption{The distribution of the configurations of $\Omega^2$ according with the type of limit cycle that belongs, for a $4\times 4$ periodic lattice. The first column indicates the period-$T$, the second (resp. third, fourth and fifth) one, the total number of configurations belonging in a  period-$T$ limit cycle of type S-I (resp. S-II, S-III and AS). Finally, the column $|P_T|$ indicates the size of the set $P_T$.}\label{TableConfigs}
\end{table}

\begin{table}[h!]
\scriptsize
\begin{tabular}{|c||c|c|c||c||c|}
\hline
$T$ & { $n_{\rm SI}(T)$} & { $n_{\rm SII}(T)$} & { $n_{\rm SIII}(T)$} & { $n_{\rm AS}(T)$} & $n(T)$\\ \hline\hline
1 & 1,156 & 0 & 0 & 0 & 1,156\\ \hline
2 & 0 & 0 & 667,590 & 0 & 667,590\\ \hline
3 & 1,376 & 0 & 0 & 256 & 1,632\\ \hline
4 & 0 & 3614 & 5,139,064 & 12,096,102 & 17,238,780\\ \hline
5 & 384 & 0 &0  &0  & 384\\ \hline
6 &0  & 1,760 &2,536,656  & 3,342,496 & 5 ,880 ,912\\ \hline
8 & 0 & 5,344 &7,299,968  & 29,375,904 & 36,681,216\\ \hline
9 & 384 & 0 & 0 & 512 & 896\\ \hline
10 & 0 & 768 & 517,440 & 294,144 & 812,352\\ \hline
12 & 0  & 8,304 & 11,024,512 & 54,609,744 & 65,642,560\\ \hline
18 & 0 & 3,840 & 1,045,824 &7,985,152  & 9,034,816\\ \hline
20 & 0 & 960 & 864,768 & 2,684,736 & 3,550,464\\ \hline 
24 & 0 &1,152  &4,820,992  & 22,342,528 & 27,164,672\\ \hline
27 & 0 & 0 & 0 &256  & 256\\ \hline
30 & 0 & 0 &528,384  & 98,304 & 626,688\\ \hline
36 & 0 &0  & 1,417,728 & 9,260,800 & 10 ,678 ,528\\ \hline
40 & 0 & 1,920 & 657,408 & 6,160,512 & 6,819,840\\ \hline
54 & 0 & 3,456 & 0 &4,494,848  &  4,498,304\\ \hline
60 & 0 & 0 & 552,960 &1,886,208  & 2,439,168\\ \hline
72 & 0 & 0 & 657,408 & 2,252,800 & 2,910,208\\ \hline
90 & 0 & 0 & 147,456 &196,608  & 344,064\\ \hline
108 & 0 & 0 & 0 & 3,051,008 & 3,051,008\\ \hline
120 & 0 & 0 & 0 & 1,671,168 & 1,671,168\\ \hline
180 & 0 & 0 & 0 & 172,032 & 172,032\\ \hline
216 & 0 & 0 & 0 & 831,488 & 831,488\\ \hline
270 & 0 & 0 & 0 & 98,304 & 98,304\\ \hline
360 & 0 & 0 & 0 & 172,032 & 172,032\\ \hline
540 & 0 &  0& 0 & 49,152 & 49,152\\ \hline
1080 & 0 & 0 & 0 &49,152  & 49,152\\ \hline\hline
Total & 3,300 &31,118  & 37,878,158 & 163,176,246 & 201,088,822\\ \hline
\end{tabular}
\caption{The distribution of the number of limit cycles for a $4\times 4$ periodic lattice according with its type. The first column indicates the period -$T$, the second (resp. third, fourth and fifth) one, the total number of period-$T$ limit cycles of type S-I (resp. S-II, S-III and AS). Finally, the column $n(T)$ indicates the total number of period-$T$ limit cycles.}\label{TableCycles}
\end{table}

The largest number of limit cycles for a given period is obtained for $T=12$ where $n(12) = 65,642,560$. The limit cycles of period-12  also correspond with the largest number of configurations, $\nu(12)\approx 787 \times 10^6$, which is about an $18\%$ of the total phase space. On the contrary, the smallest set is $P_1$ with just 1,156 configurations ($\approx  10^{-5}\, \%$). However, Table \ref{TableCycles} shows that the smallest number  of limit cycles  are those of period 27 with 256 limit cycles.

More important, from the total $2^{32}$ configurations, a fraction of $13.33\, \%$ are  symmetric limit cycles ($ 10^{-4}\, \%$ of type S-I, $ 10^{-2}\,\%$ of type S-II and  $13.31\,\%$ of type S-III) and $86.67\, \%$ are  asymmetric (AS) limit cycles. 

 The values showed in the above tables are also summarized in Figure \ref{PlotNumberOfCycles} where it can be observed that, apparently, the quantities $\nu_{\rm SI}(T)$, $\nu_{\rm SII}(T)$, $\nu_{\rm SIII}(T)$ and $\nu_{\rm AS}(T)$ (see Figure \ref{PlotNumberOfCycles}-a) are upper bounded by $K \,T^{-1/2}$, with $K$ a constant. On the contrary the quantities $n_{\rm SI}(T)$, $n_{\rm SII}(T)$, $n_{\rm SIII}(T)$ and $n_{\rm AS}(T)$ (see Figure \ref{PlotNumberOfCycles}-b) are upper bounded by  $K\, T^{-3/2}$ with $K$ a constant. We do not know the reasons for the bound.

 Moreover, because of Theorem \ref{teo_rica}
we are able to upper bound the total number of S-I and S-II limit cycles as follows:
$$\sum_{T\geq 1} \left(n_{\rm SI}(T) +n_{\rm SII}(T) \right) < | \Omega^2_{xx}|=2^N \ll | \Omega^2_{xy}|=2^N (2^N-1).$$  
This can be noticed in both Figures \ref{PlotNumberOfCycles}.


 \begin{figure}[h]
\begin{center}
a)  \includegraphics[width = 5cm]{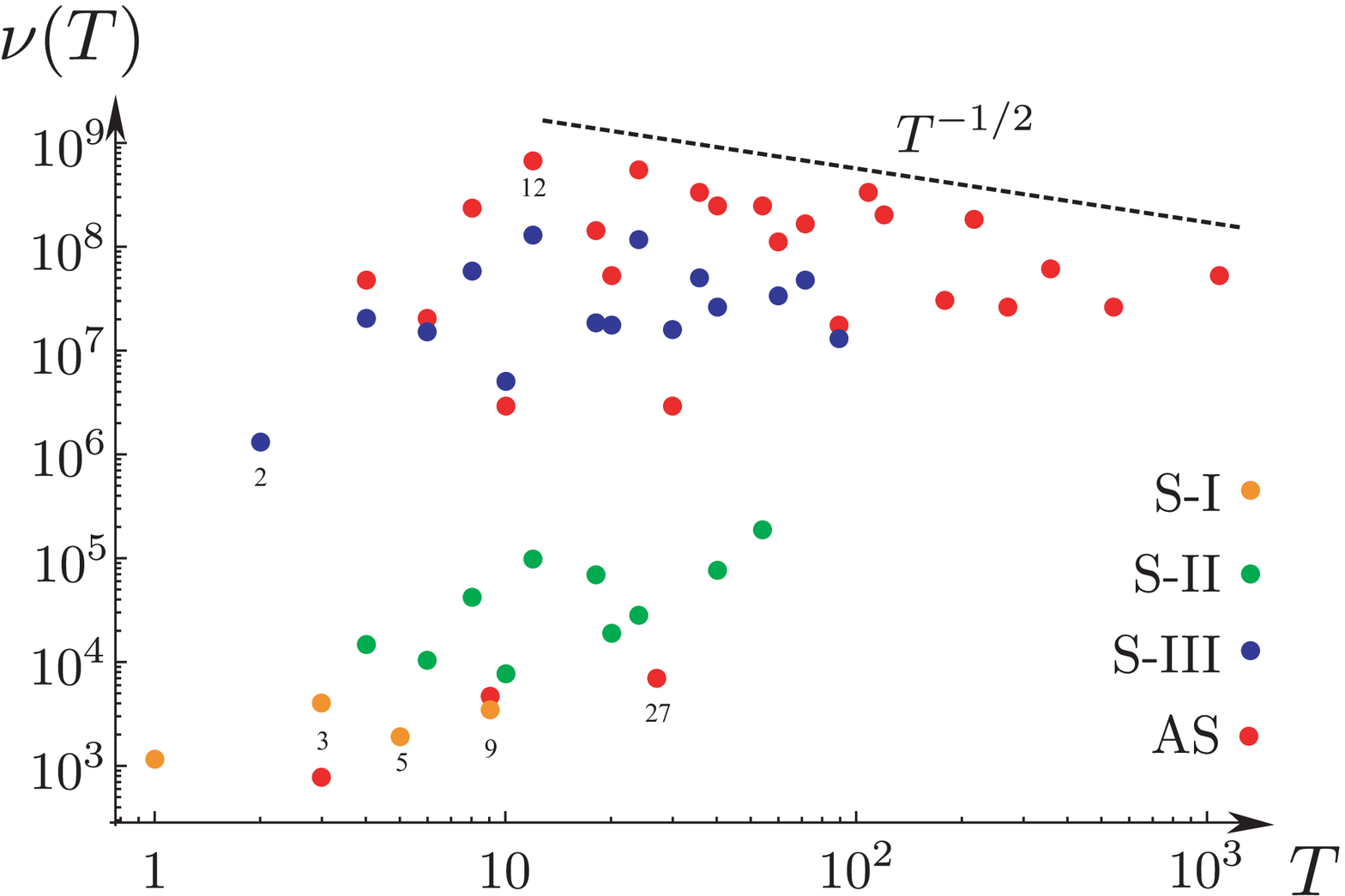}  \quad
b) \includegraphics[width = 5 cm]{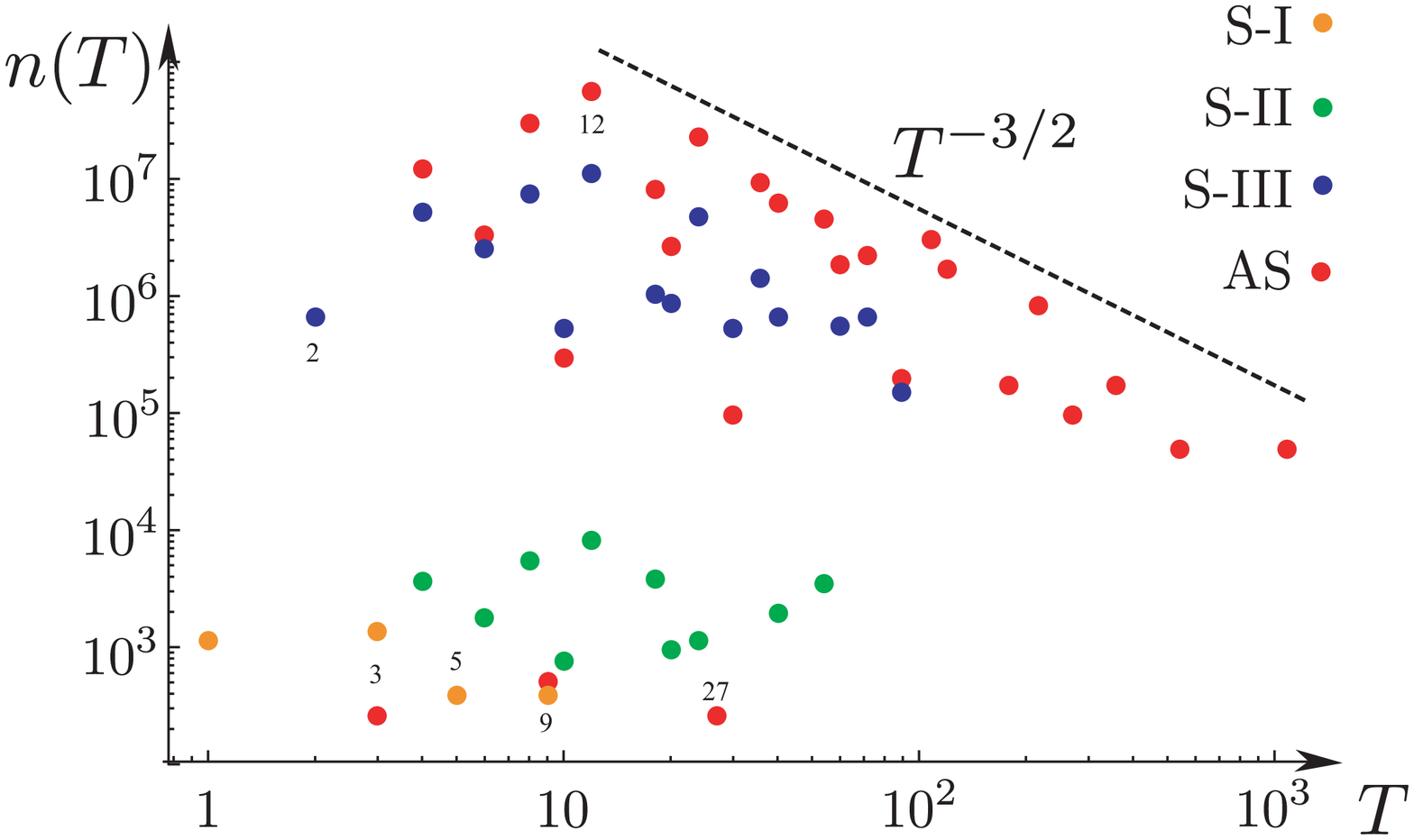} 
\end{center}
\caption{ \label{PlotNumberOfCycles}  a) Number of configurations, $\nu_{q}(T)$, and (b) number of limit cycles, $n_{q}(T)$, per period, for $q\in\{ {\rm SI},{\rm SII},{\rm SIII},{\rm AS}\}$.
}
\end{figure}

Finally, the following Figure \ref{PlotDensityOfCycles} plots the normalized density of limit cycles for each topology. To do that, we normalize $n_{\rm SI}(T)$  by $\Omega^2_{xx}=2^N$, that is 
$$\varrho_{\rm SI}(T)= \frac{n_{\rm SI}(T)}{2^N} .$$ 
In the same way we normalize
$$\varrho_{\rm SII}(T)= \frac{2 \, n_{\rm SII}(T)}{2^N}, \quad  \varrho_{\rm SIII}(T)= \frac{ n_{\rm SII}(T)}{2^N(2^N-1)},  \quad  \varrho_{\rm AS}(T)= \frac{ n_{\rm AS}(T)}{2^N(2^N-1)}.$$ 

 \begin{figure}[h]
\begin{center}
  \includegraphics[width = 7cm]{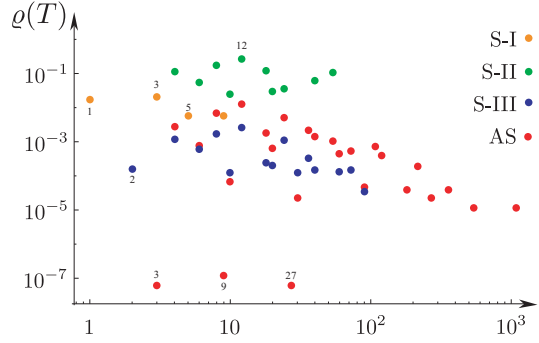}
\end{center}
\caption{ \label{PlotDensityOfCycles}  Normalized density of limit cycles per period for different topology. 
}
\end{figure}

The interest of this plot is that it shows that relative to its set, namely $\Omega_{xx}^2$, the S-I and S-II are of the same order of magnitude, than S-III and AS relative to  $\Omega_{xy}^2$.

\subsection{S-I Example: Limit cycle of odd period} \label{AppSI}
An example of an odd period limit cycle (different of a fixed point) is the following period-3 limit cycle of the form  $(x,x) \to (y,x)   \to (x,y) \to (x,x)$ (type S-I), like those of Figures \ref{Period3}a) or \ref{fig_cycletype}a), where $x$ is in red and $y$ is in blue:

\begin{equation} \label{cycleSI}
\begin{array}{ccc}
(x,x)=\left( \textcolor{red}{ \left[
\begin{array}{cccc}
 -1 & -1 & -1 & 1 \\
 -1 & -1 & 1 & 1 \\
 -1 & -1 & 1 & -1 \\
 -1 & -1 & 1 & 1 \\
\end{array}
\right]} \right. & , & \left. \textcolor{red}{ \left[
\begin{array}{cccc}
 -1 & -1 & -1 & 1 \\
 -1 & -1 & 1 & 1 \\
 -1 & -1 & 1 & -1 \\
 -1 & -1 & 1 & 1 \\
\end{array}
\right] }\right)\in C \\
(y,x)=\left(\textcolor{blue}{  \left[
\begin{array}{cccc}
 -1 & -1 & -1 & -1 \\
 -1 & -1 & -1 & -1 \\
 -1 & -1 & -1 & -1 \\
 -1 & -1 & -1 & -1 \\
\end{array}
\right] } \right. & , & \left. \textcolor{red}{ \left[
\begin{array}{cccc}
 -1 & -1 & -1 & 1 \\
 -1 & -1 & 1 & 1 \\
 -1 & -1 & 1 & -1 \\
 -1 & -1 & 1 & 1 \\
\end{array}
\right]} \right)\in B \\
(x,y)=\left( \textcolor{red}{   \left[
\begin{array}{cccc}
 -1 & -1 & -1 & 1 \\
 -1 & -1 & 1 & 1 \\
 -1 & -1 & 1 & -1 \\
 -1 & -1 & 1 & 1 \\
\end{array}
\right] } \right. & , & \left. \textcolor{blue}{ \left[
\begin{array}{cccc}
 -1 & -1 & -1 & -1 \\
 -1 & -1 & -1 & -1 \\
 -1 & -1 & -1 & -1 \\
 -1 & -1 & -1 & -1 \\
\end{array}
\right]} \right)\in D
\end{array}
\end{equation}
Observe that, by Remark \ref{rem_cyclededuction}, there not exists limit cycles of type S-I with an even period.

\subsection{S-II Example: Limit cycle of even period}\label{AppSII}
The following example is a period-4 limit cycle (type S-II) of the form $(x,x) \to (y,x) \to (y,y)  \to (x,y)  \to (x,x) $, like those of Figures \ref{fig_cycletype}b) or \ref{Period4}a), where $x$ is in red and $y$ is in blue:

\begin{equation} \label{cycleSII}
\begin{array}{ccc}
(x,x)=\left(\textcolor{red}{ \left[
\begin{array}{cccc}
 1 & 1 & 1 & 1 \\
 1 & 1 & 1 & 1 \\
 1 & 1 & 1 & 1 \\
 1 & -1 & -1 & -1 \\
\end{array}
\right]} \right. & , & \left. \textcolor{red}{ \left[
\begin{array}{cccc}
 1 & 1 & 1 & 1 \\
 1 & 1 & 1 & 1 \\
 1 & 1 & 1 & 1 \\
 1 & -1 & -1 & -1 \\
\end{array}
\right]}\right)\in C \\
(y,x)=\left( \textcolor{blue}{ \left[
\begin{array}{cccc}
 1 & 1 & 1 & 1 \\
 1 & 1 & 1 & 1 \\
 1 & 1 & 1 & 1 \\
 -1 & -1 & 1 & -1 \\
\end{array}
\right] } \right. & , & \left. \textcolor{red}{  \left[
\begin{array}{cccc}
 1 & 1 & 1 & 1 \\
 1 & 1 & 1 & 1 \\
 1 & 1 & 1 & 1 \\
 1 & -1 & -1 & -1 \\
\end{array}
\right] } \right)\in D  \\
(y,y)=\left( \textcolor{blue}{ \left[
\begin{array}{cccc}
  1 & 1 & 1 & 1 \\
 1 & 1 & 1 & 1 \\
 1 & 1 & 1 & 1 \\
 -1 & -1 & 1 & -1 \\
\end{array}
\right]} \right. & , & \left. \textcolor{blue}{ \left[
\begin{array}{cccc}
  1 & 1 & 1 & 1 \\
 1 & 1 & 1 & 1 \\
 1 & 1 & 1 & 1 \\
 -1 & -1 & 1 & -1 \\
\end{array}
\right]} \right)\in C  \\
(x,y)=\left( \textcolor{red}{ \left[
\begin{array}{cccc}
  1 & 1 & 1 & 1 \\
 1 & 1 & 1 & 1 \\
 1 & 1 & 1 & 1 \\
 1 & -1 & -1 & -1 \\
\end{array}
\right]} \right. & , & \left. \textcolor{blue}{ \left[
\begin{array}{cccc}
 1 & 1 & 1 & 1 \\
 1 & 1 & 1 & 1 \\
 1 & 1 & 1 & 1 \\
 -1 & -1 & 1 & -1 \\
\end{array}
\right]} \right)\in D  \\
\end{array}
\end{equation} 

Observe that, by Remark \ref{rem_cyclededuction}, there not exists limit cycles of type S-II with an odd period.

\subsection{S-III Example: Limit cycle of even period} \label{AppSIII}
An example of an even period limit cycle (different of a period-2 limit cycle) is the following period-4 limit cycle of the form $(x,y) \to (z,x)\to (x,z)  \to (y,x)  \to (x,y) $ (type S-III), like those of the Figures \ref{fig_cycletype}c) or \ref{Period4}b), where $x$ is in red, $y$ is in blue and $z$ is in green:

$$\begin{array}{ccc}
(x,y)=\left(
\textcolor{red}{ \left[\begin{array}{cccc} 1 & 1 & 1 & 1\\ 1 & -1 & -1 & -1\\ 1 & 1 & -1 & 1\\ 1 & 1 & 1 & -1 \end{array}\right]} \right. & , & \left.\textcolor{blue}{\left[\begin{array}{cccc} 1 & 1 & 1 & 1\\ 1 & 1 & 1 & 1\\ 1 & -1 & 1 & -1\\ -1 & 1 & -1 & 1 \end{array}\right]}\right)\in D\\
(z,x)=\left(\textcolor{green}{\left[\begin{array}{cccc} 1 & 1 & 1 & -1\\ -1 & 1 & 1 & 1\\ 1 & 1 & 1 & -1\\ -1 & 1 & 1 & 1 \end{array}\right]}\right. & , & \left.\textcolor{red}{ \left[\begin{array}{cccc} 1 & 1 & 1 & 1\\ 1 & -1 & -1 & -1\\ 1 & 1 & -1 & 1\\ 1 & 1 & 1 & -1 \end{array}\right]}\right)\in B\\
(x,z)=\left(\textcolor{red}{\left[\begin{array}{cccc} 1 & 1 & 1 & 1\\ 1 & -1 & -1 & -1\\ 1 & 1 & -1 & 1\\ 1 & 1 & 1 & -1 \end{array}\right]}\right. & , & \left.\textcolor{green}{\left[\begin{array}{cccc} 1 & 1 & 1 & -1\\ -1 & 1 & 1 & 1\\ 1 & 1 & 1 & -1\\ -1 & 1 & 1 & 1 \end{array}\right]}\right)\in D\\
(y,x)=\left(\textcolor{blue}{\left[\begin{array}{cccc} 1 & 1 & 1 & 1\\ 1 & 1 & 1 & 1\\ 1 & -1 & 1 & -1\\ -1 & 1 & -1 & 1 \end{array}\right]}\right. & , & \left.\textcolor{red}{\left[\begin{array}{cccc} 1 & 1 & 1 & 1\\ 1 & -1 & -1 & -1\\ 1 & 1 & -1 & 1\\ 1 & 1 & 1 & -1 \end{array}\right]}\right)\in B
\end{array}
$$
Observe that, by Remark \ref{rem_cyclededuction}, there not exists limit cycles of type S-III with an odd period.

\subsection{AS Example 1: Limit cycle of odd period}
The following is a period-3 limit cycle of the form $(x,y) \to (z,x)  \to (y,z)  \to (x,y) $ (type AS), like those of the Figures \ref{Period3}b) or \ref{fig_cycletype}d), where $x$ is in red, $y$ is in blue and $z$ is in green:

$$
\begin{array}{ccc}
(x,y)=\left(\textcolor{red}{ \left[
\begin{array}{cccc}
 -1 & 1 & -1 & -1 \\
 -1 & -1 & 1 & -1 \\
 -1 & -1 & 1 & -1 \\
 -1 & 1 & -1 & -1 \\
\end{array}
\right]} \right. & , & \left. \textcolor{blue}{ \left[
\begin{array}{cccc}
 -1 & -1 & -1 & -1 \\
 -1 & 1 & 1 & -1 \\
 -1 & -1 & -1 & -1 \\
 -1 & 1 & 1 & -1 \\
\end{array}
\right]}\right)\in D  \\
(z,x)=\left( \textcolor{green}{ \left[
\begin{array}{cccc}
 -1 & -1 & 1 & -1 \\
 -1 & -1 & 1 & -1 \\
 -1 & 1 & -1 & -1 \\
 -1 & 1 & -1 & -1 \\
\end{array}
\right] } \right. & , & \left. \textcolor{red}{  \left[
\begin{array}{cccc}
 -1 & 1 & -1 & -1 \\
 -1 & -1 & 1 & -1 \\
 -1 & -1 & 1 & -1 \\
 -1 & 1 & -1 & -1 \\
\end{array}
\right] } \right)\in D  \\
(y,z)=\left( \textcolor{blue}{ \left[
\begin{array}{cccc}
 -1 & -1 & -1 & -1 \\
 -1 & 1 & 1 & -1 \\
 -1 & -1 & -1 & -1 \\
 -1 & 1 & 1 & -1 \\
\end{array}
\right]} \right. & , & \left. \textcolor{green}{ \left[
\begin{array}{cccc}
 -1 & -1 & 1 & -1 \\
 -1 & -1 & 1 & -1 \\
 -1 & 1 & -1 & -1 \\
 -1 & 1 & -1 & -1 \\
\end{array}
\right]} \right)\in D  \\
\end{array}
$$

\subsection{AS Example 2: Limit cycle of even period}
The following ia a period-4 limit cycle of the form $(x,y) \to (z,x)\to (u,z)  \to (y,u)  \to (x,y) $ (type AS), like those of the Figures \ref{fig_cycletype}d) or \ref{Period4}c), where $x$ is in red, $y$ is in blue, $z$ is in green and $u$ is in black:

$$
\begin{array}{ccc}
(x,y)=\left(\textcolor{red}{ \left[
\begin{array}{cccc}
 1 & 1 &1 & 1 \\
 1 & 1 & 1 & -1 \\
 -1 & 1 & -1 & 1 \\
 -1 & 1 & 1 & -1 \\
\end{array}
\right]} \right. & , & \left. \textcolor{blue}{ \left[
\begin{array}{cccc}
 1 & 1 & 1 & 1 \\
 1 & 1 & 1 & 1 \\
 1 & 1 & 1 & 1 \\
 1 & -1 & 1 & -1 \\
\end{array}
\right]}\right)\in D  \\
(z,x)=\left( \textcolor{green}{ \left[
\begin{array}{cccc}
 1 & 1 & 1 & -1 \\
 -1 & 1 & -1 & 1 \\
 1 & -1 & 1 & 1 \\
 -1 & -1 & -1 & -1 \\
\end{array}
\right] } \right. & , & \left. \textcolor{red}{  \left[
\begin{array}{cccc}
 1 & 1 &1 & 1 \\
 1 & 1 & 1 & -1 \\
 -1 & 1 & -1 & 1 \\
 -1 & 1 & 1 & -1 \\
\end{array}
\right] } \right)\in D  \\
(u,z)=\left( { \left[
\begin{array}{cccc}
 1 & 1 &1 & 1 \\
 1 & 1 & 1 & -1 \\
 -1 & 1 & -1 & 1 \\
 1 & 1 & -1 & -1 \\
\end{array}
\right]} \right. & , & \left. \textcolor{green}{ \left[
\begin{array}{cccc}
 1 & 1 & 1 & -1 \\
 -1 & 1 & -1 & 1 \\
 1 & -1 & 1 & 1 \\
 -1 & -1 & -1 & -1 \\
 \end{array}
\right]} \right)\in D  \\
(y,u)=\left( \textcolor{blue}{ \left[
\begin{array}{cccc}
 1 & 1 &1 & 1 \\
 1 & 1 & 1 &1 \\
 1 & 1 & 1 & 1 \\
 1 & -1 & 1 & -1 \\
\end{array}
\right]} \right. & , & \left. { \left[
\begin{array}{cccc}
  1 & 1 &1 & 1 \\
 1 & 1 & 1 & -1 \\
 -1 & 1 & -1 & 1 \\
 1 & 1 & -1 & -1 \\
 \end{array}
\right]} \right)\in D  \\
\end{array}
$$

\section*{Acknowledgment}
Work partially supported by FONDECYT Iniciaci\'on 11150827 (M.M-M.). S.R. thanks the Gaspard Monge Visiting Professor Program of \'Ecole Polytechnique (France). F.U. thanks FONDECYT (Chile) for financial support through Postdoctoral N$^\circ$ $3180227$. Finally, the authors thank Fondequip AIC-34. Powered@NLHPC: This research was partially supported by the supercomputing infrastructure of the NLHPC (ECM-02).

\end{document}